\documentclass[10pt,journal, a4paper]{IEEEtran}
\normalsize
\usepackage{graphicx}
\usepackage{multicol}

\usepackage{lipsum}
\usepackage{diagbox}
\usepackage{blkarray}

\usepackage{hyperref}
\usepackage{float}
\usepackage{url}
\usepackage[utf8]{inputenc} 
\usepackage{array}
\usepackage[utf8]{inputenc} 
\usepackage[T1]{fontenc}
\usepackage{amsmath}
\usepackage{stfloats}
\usepackage{amsmath}[ntheorem]
\usepackage{mathtools}
\usepackage{amsfonts}
\usepackage{tabstackengine}
\usepackage{amssymb}
\usepackage{amsthm}
\usepackage{enumerate}
\usepackage{cite}
\usepackage{calc}
\usepackage{nicematrix}
\usepackage{color}
\usepackage{epsfig}
\usepackage{setspace}
\usepackage{pstricks}
\usepackage{cancel}
\usepackage{multirow}
\usepackage{mathrsfs}
\usepackage{algorithm}
\usepackage{algpseudocode}
\usepackage{multirow}
\usepackage{romannum}
\usepackage{diagbox}
\algnewcommand{\algorithmicor}{\textbf{ or }}
\algnewcommand{\OR}{\algorithmicor}
\usepackage{booktabs,makecell}
\usepackage{mathrsfs}
\usepackage{enumitem}
\newtheorem{defn}{Definition}
\newtheorem{thm}{Theorem}

\newtheorem{prop}{Proposition}

\newtheorem{remark}{Remark}

\theoremstyle{definition}
\newtheorem{example}{Example}
\graphicspath{{./images/}}
\def\Ddots{\mathinner{\mkern1mu\raise\p@
		\vbox{\kern7\p@\hbox{.}}\mkern2mu
		\raise4\p@\hbox{.}\mkern2mu\raise7\p@\hbox{.}\mkern1mu}}
\makeatother

\begin{document}
	\pagenumbering{arabic}
	\title{Combinatorial Multi-Access Coded Caching with Private Caches}
	\author{\IEEEauthorblockN{Dhruv Pratap Singh, Anjana A. Mahesh and B. Sundar Rajan}\\
		\IEEEauthorblockA{Department of Electrical Communication Engineering, Indian Institute of Science, Bengaluru}\\E-mail: \{dhruvpratap,anjanamahesh,bsrajan\}@iisc.ac.in}
	\maketitle
	\begin{abstract}
		We consider a variant of the coded caching problem where users connect to two types of caches, called private and access caches. The problem setting consists of a server with a library of files and a set of access caches. Each user, equipped with a private cache, connects to a distinct $r-$subset of the access caches. The server populates both types of caches with files in uncoded format. For this setting, we provide an achievable scheme and derive a lower bound on the number of transmissions for this scheme. We also present a lower and upper bound for the optimal worst-case rate under uncoded placement for this setting using the rates of the Maddah-Ali--Niesen scheme for dedicated and combinatorial multi-access coded caching settings, respectively. Further, we derive a lower bound on the optimal worst-case rate for any general placement policy using cut-set arguments. We also provide numerical plots comparing the rate of the proposed achievability scheme with the above bounds, from which it can be observed that the proposed scheme approaches the lower bound when the amount of memory accessed by a user is large. Finally, we discuss the optimality w.r.t worst-case rate when the system has four access caches. 
	\end{abstract}
	\begin{IEEEkeywords}
		Coded Caching, Combinatorial Multi-Access Network, Index Coding, Cut-Set Bound
	\end{IEEEkeywords}
	\section{Introduction}
	\textit{Coded caching} is a spectrum-sharing technique for caching systems, introduced by Maddah-Ali and Niesen in their landmark paper\cite{MAN}, that helps in reducing network traffic during peak hours. It operates in two phases: the \textit{placement phase} and the \textit{delivery phase}. During the placement phase, which occurs when the network load is low, the cache memories in the system are populated with contents in either coded\cite{CFL},\cite{J},\cite{AG} or uncoded fashion\cite{WTP1},\cite{YMA}, while adhering to the memory constraint. During the delivery phase, which commences after all users make their demands known, the server seeks to satisfy the demands of all the users with a minimum number of transmissions. The objective of a coded caching problem is to jointly design a placement and a delivery scheme that minimizes the number of file transmissions required. The scheme introduced by Maddah-Ali and Niesen\cite{MAN}, referred to as the MAN scheme, addressed the dedicated coded caching problem where a central server having $N$ files of equal length connects to $K$ users via a shared error-free link. Each user in this network possessed a dedicated cache of size $M$ ($M\leq N$) files. The MAN scheme was proven optimal \cite{WTP} under uncoded placement in the $N\geq K$ regime, where the optimality is w.r.t minimizing the rate of transmission, i.e., the load of the shared link normalized by the file size, in the delivery phase.
	
	Coded caching was studied for various other settings like decentralized placement\cite{MUD}, shared caches\cite{PUE},\cite{IZY}, hierarchical networks \cite{KNAD}, with secure delivery\cite{STC}, with privacy\cite{RPKP}, and many more. In the shared cache setting, the cache memories are not present at the users, but are shared among multiple users. Another setting in which the cache memories are not private to the users was the multi-access coded caching setting \cite{HKD,SPE}, where the cache memories are present at multiple access points in the system and not at the users. Each user could  connect to  multiple access points (caches) as well as receive broadcast transmissions from the server. Papers \cite{PD, BE} studied the multi-access coded caching setting with a combinatorial connectivity imposed between the access points and the users, and hence the setting in these papers is called the combinatorial multi-access coded caching.  
	
	This work considers an extension of the combinatorial multi-access coded caching setting where the users not only connect to the cache memories at the access points but also are endowed with their own private cache memories. Previous works in literature that considered users with access to two different types of caches, one shared between multiple users and the other private to a user, includes \cite{MR},\cite{PNR},\cite{PNR1}. The difference between these settings and the one in this paper is that, in all of \cite{MR},\cite{PNR},\cite{PNR1}, a user has access to only one access cache in addition to its private cache, whereas in this paper, each user has access to the cache memories at multiple access points as well as its private cache. 
	
	We consider a model where a server with $N$ files connects to $K$ users and $\Lambda$ access caches via an error-free wireless link. Each user connects to a distinct $r-$subset of the $\Lambda$ caches. However, unlike \cite{PD},\cite{BE}, each user also has a private cache. This network is a generalization of the multi-access combinatorial \cite{PD} and dedicated \cite{MAN} caching networks. We refer to this network as the combinatorial multi-access plus private (CMAP) coded caching setting. This setting is akin to cache-enabled users (like cell phones) connecting to several access points in an environment. To the best of our knowledge, this is the first work that studies coded caching for this system.
	
	\subsection{Our Contributions}
	We introduce a novel combinatorial network architecture incorporating access and private caches. This network can be viewed as the generalization of the dedicated caching network, introduced in \cite{MAN}, and the combinatorial multi-access caching network, introduced in \cite{PD}. Our contributions are presented below:
	\begin{itemize}
		\item The optimal rate for the CMAP system, under uncoded placement, is lower and upper bounded using the rates of the schemes presented in \cite{MAN} and \cite{PD}.
		\item The optimal rate for the CMAP system, under any general placement, is lower bounded using cut-set arguments\cite{CT}.
		\item A centralized coded caching scheme is proposed for the CMAP setting when the private cache memory takes a particular value. It is shown that the proposed scheme reverts to the combinatorial multi-access coded caching scheme in \cite{PD} when the private cache memories at the users are of size zero.
		\item A lower bound on the number of transmissions made during the delivery phase for the placement policy presented in this paper is derived using index coding arguments.
		\item Numerical  plots are given to compare the rate attained by the achievability scheme with the upper and lower bounds proposed in this paper. 
		\item A placement policy applicable for any general value of the private cache memory size and a discussion on the optimality for the CMAP coded caching system having four access caches are also presented.
	\end{itemize}
	%	In this work, we bound the optimal worst-case rate under uncoded placement for the CMAP coded caching setting using the rates achieved by the schemes in \cite{MAN} and \cite{PD}. We provide an achievability scheme for the CMAP coded caching setting. Given the placement used for the achievability scheme, we provide a lower bound on the number of transmissions using index coding techniques. 
	
	\textit{Organization of the paper}: Section \ref{systemmodel} introduces the system model and the preliminaries needed for proofs later in the paper. The main results of this paper are presented in section \ref{mainresults}, the proofs of which are provided in the subsequent section \ref{proofs}. In section \ref{numericalcomparison}, we provide numerical plots for comparison of the rate attained by the achievability scheme with the bounds derived in section \ref{mainresults}. A general placement policy followed by a discussion on the optimality for the CMAP coded caching system having four access caches is provided in section \ref{optimalityforlambda4}. Finally, we conclude the paper in section \ref{conclusions}. 
	
	\textit{Notation:} The set $\{a,a+1,\cdots,b\}$ where $b\geq a$ is denoted by $[a,b]$, for some $a,b\in\mathbb{Z}^+$, where $\mathbb{Z}^+$ is the set of all non-negative integers. The cardinality of a set $A$ is denoted as $|A|$. The finite field with $q$ elements is denoted as $\mathbb{F}_q$. The smallest integer not less than $a$ is denoted by $\lceil a\rceil$, while $\lfloor a\rfloor$ denotes the largest integer not greater than $a$. The binomial coefficient $\frac{n!}{k!(n-k)!}$ is denoted as $\binom{n}{k}$ and we assume $\binom{n}{k}=0$ if $n < 0, k < 0$ or $n<k$. We use the $\oplus$ symbol to denote the bitwise XOR operation.
	\section{System Model and Preliminaries}
	\label{systemmodel}
	In this section, we first introduce the system model. After that, we discuss the MAN scheme for dedicated and combinatorial multi-access coded caching systems and revisit some results from index coding that are used in this work.
	\subsection{System Model}
	\begin{figure}
		\centering
		\includegraphics[scale=0.7]{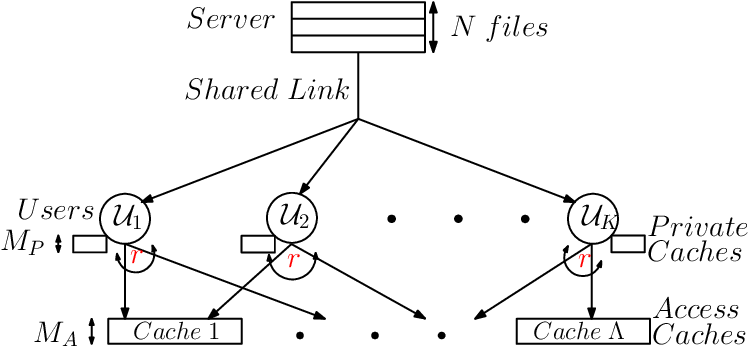}
		\caption{The $(\Lambda,r,M_a,M_p,N)-$CMAP Coded Caching System.}
		\label{fig1}
	\end{figure}
	Consider the system model as shown in Fig. \ref{fig1}. The central server has $N$ files of $B$ bits each, denoted by $W_1, W_2,\cdots, W_N$. The server connects to $K$ users via an error-free wireless broadcast link such that $N\geq K$. The system has $\Lambda$ caches, each capable of storing $M_a\leq N$ files, that are accessed by multiple users via error-free infinite-capacity wireless links. We refer to these caches as the access caches. %Every distinct $r-$subset of the $\Lambda$ caches is accessed by a single user having a private cache capable of storing $M_p\leq N$ files, where $r$ is the access degree. Since a single user is accessing a distinct $r-$subset of the access caches, users are indexed by the set of access caches they connect to, and there are $K=\binom{\Lambda}{r}$ users. We refer to this model as $(\Lambda,r, M_a, M_p, N)-$CMAP coded caching setting.
	Every distinct $r-$subset of the $\Lambda$ caches is accessed by a single user, where $r$ is the access degree. Users are indexed by the subset of access caches they connect to, resulting in a total of $K=\binom{\Lambda}{r}$ users. Each user has a private cache capable of storing $M_p\leq N$ files. The memory pairs $(M_a, M_p)$ that are of interest satisfy the constraint $M_a+M_p < N$. This model is referred to as the $(\Lambda,r, M_a, M_p, N)-$CMAP coded caching setting. The system operates in two phases:
	\begin{enumerate}
		\item Placement phase: The server populates the private and the access caches with parts of the files in either coded or uncoded fashion while adhering to their respective memory constraints. The server employs a caching mechanism wherein files are divided into subfiles. In this paper, these subfiles are further broken down into mini-subfiles, which are then stored in the private cache of the users. Meanwhile, the access caches are populated directly with the subfiles. The number of mini-subfiles each file is divided into is called the subpacketization level. The content of the access cache $a$ is denoted by $Z_a$ while the contents of the private cache of the user $\mathcal{U}$ is denoted by $Z^p_\mathcal{U}$. For a particular subfile that a user accesses from an access cache, we assume that all the corresponding mini-subfiles are available to the user. The set of mini-subfiles that a user $\mathcal{U}$ has, from the access caches it connects to as well as from its private cache is denoted as $\mathcal{Z}_\mathcal{U}$ . 
		\item Delivery phase: Each user $\mathcal{U}$ demands one of the $N$ files from the server. The demands of all the users are encapsulated in the demand vector $\mathbf{d} = (d_\mathcal{U}:\mathcal{U}\subseteq[1,\Lambda],|\mathcal{U}|=r)$. After the demand vector is known, the server aims to satisfy the demands of all the users with the minimum number of transmissions. Each transmission consists of a coded combination of mini-subfiles, achieved through bitwise XOR operations. As a result, each transmission contains the same number of bits as there are in one mini-subfile. The rate $R$ is defined as the number of transmissions made by the server in the unit of files. Since each transmission is of the size of one mini-subfile, the rate $R$ is defined as the number of transmissions made by the server normalized by the subpacketization. The maximum number of transmissions occurs when each user demands a different file. This results in the worst-case rate.
		\begin{defn}
			\label{defn1}
			Consider the $(\Lambda,r,M_a,M_p,N)-$CMAP coded caching setting. We say that the triplet $(M_a, M_p, R)$ is achievable if there exists a coded caching scheme that achieves the rate $R$ with the memory pair $(M_a, M_p)$ for a large enough file size. We define the optimal worst-case rate for the $(\Lambda,r,M_a,M_p,N)-$CMAP coded caching setting as			
			\begin{align*}
				R^{\textasteriskcentered}(M_a,M_p)=\inf\{R:(M_a,M_p,R)\text{ is achievable}\}.
			\end{align*}
		\end{defn}
		The objective is to design joint placement and delivery policies such that $R^{\textasteriskcentered}(M_a, M_p)$ is achieved.
	\end{enumerate}
	\subsection{MAN Scheme}
	The MAN scheme \cite{MAN} is defined for the dedicated caching network where $K$ users connect to a central server having $N$ files. Every user connects to a dedicated cache capable of storing $M\leq N$ files. The delivery phase starts when the server is informed of the demand vector $\mathbf{d}=(d_1,d_2,\cdots,d_K)$, where $d_k$ is the index of the file demanded by the $k^{\text{th}}$ user.
	\begin{enumerate}
		\item Placement Phase: Each file is divided into $\binom{K}{t}$ subfiles as $W_n=\{W_{n,\mathcal{T}}:\mathcal{T}\subseteq[1,K],|\mathcal{T}|=t\}$, where $t=\frac{KM}{N}\in\mathbb{Z}^+$. The contents of the cache connected to user $k\in[1,K]$ is $Z_k=\{W_{n,\mathcal{T}}:i\in\mathcal{T},\mathcal{T}\subseteq[1,K],|\mathcal{T}|=t, \forall n\in[1,N]\}$.	
		\item Delivery Phase: The server makes the broadcast transmission $T_{\mathcal{S}}$ for every subset $\mathcal{S}$ of $[1,K]$, where $|\mathcal{S}|=t+1$ and $T_{\mathcal{S}}=\bigoplus\limits_{s\in \mathcal{S}} W_{{d_s},\mathcal{S}\setminus\{s\}}$.
		\item Rate: Each file is divided into $\binom{K}{t}$ subfiles, and a transmission is made for every $(t+1)$ subset of the users; we have the rate, shown optimal under uncoded placement for $N\geq K$ in \cite{YMA}, as $R^{\textasteriskcentered}_D(M)=\frac{\binom{K}{t+1}}{\binom{K}{t}}$.
	\end{enumerate}
	\subsection{MAN Scheme for Combinatorial Multi-Access Coded Caching (CMACC) Network}
	Consider a combinatorial multi-access setting with $N$ files, $\Lambda$ caches, each of memory $M\leq N$ files, and $K=\binom{\Lambda}{r}$ users, each accessing a distinct $r-$subset of the $\Lambda$ caches. The contents of the cache $\lambda\in[1,\Lambda],$ are given by $Z_\lambda=\{W_{n,\mathcal{T}}:\lambda\in\mathcal{T},\mathcal{T}\subseteq[1,\Lambda],|\mathcal{T}|=t, \forall n\in[1,N]\}$, where $t=\frac{\Lambda M}{N}\in\mathbb{Z}^+$. Once a demand vector $\mathbf{d}=(d_{\mathcal{U}}:\mathcal{U}\subseteq[1,\Lambda],|\mathcal{U}|=r)$ is revealed, the server transmits $T_{\mathcal{S}}=\bigoplus\limits_{\mathcal{U}\subseteq \mathcal{S}, |\mathcal{U}|=r} W_{{d_{\mathcal{U}},\mathcal{S}\setminus \mathcal{U}}},\forall\mathcal{S}\subseteq[1,\Lambda],|\mathcal{S}|=(t+r)$. This scheme\cite{PD}, shown to be optimal under uncoded placement for $N\geq K$ in \cite{BE}, results in a rate $R^{\textasteriskcentered}_{CMACC}=\frac{\binom{\Lambda}{t+r}}{\binom{\Lambda}{t}}$.
	
	\subsection{Index Coding Preliminaries}
	The index coding problem (ICP) with side information\cite{BK},\cite{YBJK} involves a single source having $n$ messages $x_1,x_2,\cdots,x_n: x_i\in \mathbb{F}_q,\forall i\in[1,n],$ broadcasting to a set of $K$ receivers, $R_1, R_2,\cdots, R_K$. A receiver $R_i$, $i \in [1, K]$, possesses $\{x_j:j\in\mathcal{X}_i\}$, where $\mathcal{X}_i \subseteq [1,n]$ is the index set of messages belonging to the side information of receiver $R_i$. Further, each receiver $R_i$ is interested in receiving a message $x_{f(i)}$, where $f:[1,K]\rightarrow[1,n],$ and $f(i)\not\in\mathcal{X}_i$. For an ICP $\mathcal{I}$, the generalized independence number $\alpha(\mathcal{I})$ was defined in\cite{DSC} as follows: Define the set $\mathcal{Y}_i=[1,n]\setminus\left(\{f(i)\}\cup\mathcal{X}_i\right)$ for each receiver $R_i$. Define $\mathcal{J}(\mathcal{I})=\bigcup\limits_{i\in[1,K]} \{\{f(i)\}\cup Y_i:Y_i\subseteq\mathcal{Y}_i\}$. A subset $H$ of $[1,n]$ is called a generalized independent set in $\mathcal{I}$ if every subset of $H$ belongs in $\mathcal{J}(\mathcal{I})$. The generalized independent set having the largest cardinality in $\mathcal{I}$ is called the maximal generalized independent set, and its cardinality, denoted by $\alpha(\mathcal{I})$, is called the generalized independence number. It was shown in \cite{KTR} that $\alpha(\mathcal{I})$ lower bounds the number of scalar linear transmissions required to solve the ICP $\mathcal{I}$. For a given placement scheme and a given demand vector, the delivery phase of the coded caching problem can be formulated as an ICP; hence, its corresponding generalized independence number lower bounds the number of transmissions in the delivery phase required to satisfy the demands of all the users in the coded caching problem.
	\section{Main Results}
	\label{mainresults}
	In this section, we present the main results in this paper. Proposition \ref{prop1} gives lower and upper bounds on the optimal rate, under uncoded placement, for the CMAP network described in Section \ref{systemmodel}. For the same setting, Theorem \ref{thm1} presents a lower bound on the optimal worst-case rate described in Definition \ref{defn1} and Theorem \ref{thm2} presents an achievable rate. When the cache placement is done according to the placement policy proposed in this paper, which is presented later in section \ref{achievability}, Theorem \ref{thm3} provides a lower bound on the number of transmissions required in the delivery phase.
	\begin{prop}
		\label{prop1}
		For a $(\Lambda,r, M_a, M_p, N)-$CMAP coded caching system, the optimal worst-case rate $R_{UC}^{\textasteriskcentered}(M_a, M_p)$ under uncoded placement  is bounded as:
		\begin{equation*}
			R^{\textasteriskcentered}_{D}(rM_a+M_p)\leq R_{UC}^{\textasteriskcentered}(M_a,M_p)\leq R^{\textasteriskcentered}_{CMACC}(M_a+\frac{M_p}{r}),
		\end{equation*}
		where, $R^{\textasteriskcentered}_{D}(M)$ is the rate achieved by MAN scheme \cite{MAN} and $R^{\textasteriskcentered}_{CMACC}(M)$ is the rate achieved by MAN scheme for CMACC network\cite{PD}.
	\end{prop}
	\begin{proof}
		Consider a $(\Lambda,r,M_a,M_p,N)-$CMAP coded caching system. Each user in this system connects to $r$ access caches, each of which is capable of storing $M_a$ files. In addition to this, the user also has a private cache of storage capacity $M_p$ files. Hence, the total memory accessed by each user is $rM_a+M_p$. For a fair comparison, the total memory accessed by each user is kept the same in all the three settings under consideration, namely the CMAP coded caching setting, the combinatorial multi-access network, and the dedicated caching network. We will calculate the size of the caches in the combinatorial multi-access network first, followed by the calculation of the size of the cache memories in the dedicated caching network. In the multi-access coded caching network, each user connects to $r$ caches. If every cache is of size $M_{CMACC}$, the total memory accessed by the user will be $rM_{CMACC}$. Since the total memory accessed by a user is $rM_a+M_p$, we have $M_{CMACC}=M_a+\frac{M_p}{r}$. In the dedicated caching network, each user connects to a cache of size $M_D$ which implies $M_D=rM_a+M_p$. We will now prove the inequality given above.
		
		Let $Z^{\textasteriskcentered}$ and $D^{\textasteriskcentered}$ be the placement and delivery policy that results in $R_{UC}^{\textasteriskcentered}(M_a,M_p)$. Here, the contents accessible to each user $\mathcal{U}$ can be written as $\mathcal{Z}_{\mathcal{U}}=\Big\{\{\bigcup\limits_{i\in\mathcal{U}}Z_i\}\cup Z^p_{\mathcal{U}}\Big\}$. In a dedicated cache network with $K=\binom{\Lambda}{r}$ users, each having a cache of size $rM_a+M_p$ files, it is possible to follow a placement such that the contents available to user $k\in[1, K]$ is the same as $\Big\{\{\bigcup\limits_{i\in\mathcal{U}}Z_i\}\cup Z^p_{\mathcal{U}}\Big\}$ where $\mathcal{U}$ is the $k^{\text{th}}$ user, when the user-index sets are arranged lexicographically. Thus, we can conclude that by following the delivery policy $D^{\textasteriskcentered}$ in the dedicated caching network, we achieve a rate of $R_{UC}^{\textasteriskcentered}(M_a, M_p)$. Hence, we have $R^{\textasteriskcentered}_{D}(rM_a+M_p)\leq R_{UC}^{\textasteriskcentered}(M_a,M_p)$.
		
		Consider a combinatorial multi-access coded caching network with $\Lambda$ caches, each of memory $M=M_a+\frac{M_p}{r}$ files, achieving the rate $R^{\textasteriskcentered}_{CMACC}(M)$. The contents of a cache $i\in[1,\Lambda]$ can be written as $Z_i=Z_{a_i}\cup Z_{p_i}$, where $|Z_{a_i}|=M_a,\text{ and, }|Z_{p_i}|=\frac{M_p}{r}$. For the CMAP setting, if we populate the $i^{\text{th}}$ access cache as $Z_i=Z_{a_i},i\in[1,\Lambda],$ and the content of the private cache of user $\mathcal{U}$ as $Z^p_{\mathcal{U}}=\bigcup\limits_{i\in\mathcal{U}} Z_{p_i}$, following the delivery policy of \cite{PD}, we obtain a rate of $R^{\textasteriskcentered}_{CMACC}(M_a+\frac{M_p}{r})$. Hence, $R_{UC}^{\textasteriskcentered}(M_a,M_p)\leq R^{\textasteriskcentered}_{CMACC}(M_a+\frac{M_p}{r})$.	
	\end{proof}We have characterized the bounds on the optimal worst-case rate under uncoded placement for the CMAP coded caching system. Now, we provide a lower bound on the optimal worst-case rate under any general placement for the CMAP coded caching system.
	\begin{thm}
		\label{thm1}
		For a $(\Lambda,r,M_a,M_p,N)-$CMAP coded caching system, the worst-case rate is lower bounded as
		\begin{align}
			R^{\textasteriskcentered}(M_a,M_p)\geq\max\limits_{s\in\{1,2,\cdots,\min(K,N)\}} \left(s-\frac{qM_a+sM_p}{\left\lfloor\frac{N}{s}\right\rfloor}\right),
		\end{align}where $q=\min(\Lambda+r-1,\Lambda)$.
		\begin{proof}
			For $s\in\{1,2,\cdots,\min(N, K)\}$, consider the first $s$ users, given that the user-index sets are arranged in a lexicographic manner. These $s$ users will connect to the first $q=\min(s+r-1,\Lambda)$ access caches. For the demand vector where the first $s$ users request the files $W_1, W_2,\cdots, W_s$, respectively, and the remaining $K-s$ users demand arbitrary files, let the server make the transmission $T_1$. The first $s$ users decode the files $W_1, W_2,\cdots, W_s$ using $T_1$, along with the cache contents of the first $q$ access caches and their private caches. Similarly, for the demand vector where the first $s$ users request the files $W_{s+1}, W_{s+2},\cdots, W_{2s}$, and the remaining $K-s$ users make arbitrary demands, let the server make the transmission $T_2$. Using the transmission $T_2$, the contents of the first $q$ access caches and the contents of their respective private caches, the first $s$ users are able to decode the files $W_{s+1}, W_{s+2}\cdots, W_{2s}$. Continuing in this manner, the first $s$ users will be able to decode the files $W_{(\left\lfloor\frac{N}{s}\right\rfloor-1)s+1}, W_{(\left\lfloor\frac{N}{s}\right\rfloor-1)s+2},\cdots, W_{\left\lfloor\frac{N}{s}\right\rfloor s}$ using the contents of the first $q$ access caches, the contents of their respective private caches and the transmission $T_{\left\lfloor\frac{N}{s}\right\rfloor}$. The server has transmitted $\left\lfloor\frac{N}{s}\right\rfloor R^{\textasteriskcentered}(M_a, M_p)B$ bits, the first $s$ users have access to $qM_aB+sM_pB$ bits, and, using these transmissions and the cache contents, the first $s$ users have been able to decode $s\left\lfloor\frac{N}{s}\right\rfloor B$ bits. Therefore, we have,
			\begin{align*}
				&\left\lfloor\frac{N}{s}\right\rfloor R^{\textasteriskcentered}(M_a,M_p)B+qM_aB+sM_pB\geq s\left\lfloor\frac{N}{s}\right\rfloor B,\\
				&\implies R^{\textasteriskcentered}(M_a,M_p)\geq s-\frac{qM_a+sM_p}{\left\lfloor\frac{N}{s}\right\rfloor}.
			\end{align*}Maximizing over all $s\in\{1,2,\cdots,\min(N,K)\}$, we have, 
			\begin{align*}
				R^{\textasteriskcentered}(M_a,M_p)\geq \max\limits_{s\in\{1,2,\cdots,\min(N,K)\}}\left(s-\frac{qM_a+sM_p}{\left\lfloor\frac{N}{s}\right\rfloor}\right).
			\end{align*}
		\end{proof}
	\end{thm}
	We now present the achievability scheme.%Moving forward, we only consider the $ M_=\frac {N}{K}$ case.
	\begin{thm}[Achievability]
		\label{thm2}	~\\For a $(\Lambda,r,M_a,M_p=\frac{N}{\binom{\Lambda}{r}},N)-$CMAP coded caching setting, a worst-case rate
		\begin{equation}
			R=\frac{\binom{\Lambda-r-t}{r}}{\binom{t+r}{t}}+\sum\limits_{i=1}^{r-1}\frac{\binom{r}{i}\binom{\Lambda-t-r}{r-i}}{2\binom{t+i}{i}},
		\end{equation}is achievable for the subpacketization $F=\binom{\Lambda}{t}\binom{\Lambda-t}{r}$, where $t=\frac{\Lambda M_a}{N}$ and $t\in[0,\Lambda]$.
	\end{thm}
	\begin{proof}
		Section \ref{achievability} gives a scheme achieving this rate.
	\end{proof}
	%	\begin{remark}
		Note that the rate is defined only for integer values of $t$. For general $0 \leq t \leq \Lambda$, the lower convex envelope of the points in Theorem \ref{thm1} is achievable via memory sharing, as explained in section \ref{achievability}.
		\begin{thm}[Alpha Bound]
			\label{thm3}
			For a $(\Lambda,r, M_a, M_p=\frac{N}{K}, N)-$CMAP coded caching setting,  and the placement policy described in Section \ref{achievability}, the number of transmissions $T$ required to satisfy the demands of all the users is lower bounded as
			\begin{equation}
				\begin{split}
					T\geq\binom{\Lambda-t}{r}\binom{\Lambda-r+1}{t+1}-\binom{\Lambda-r+2}{t+2}+\\\frac{\left(\binom{\Lambda-t}{r}-\Lambda+r+t-2\right)\left(\binom{\Lambda-t}{r}-\Lambda+r+t-1\right)}{2}.
				\end{split}
			\end{equation}
		\end{thm}
		\begin{proof}
			The proof is provided in section \ref{alphabound}.
		\end{proof}
		\section{Achievability and Lower Bound}
		\label{proofs}
		In this section, we present the general placement and delivery scheme that achieves the rate in Theorem \ref{thm1} and provide an index coding based lower bound on the number of transmissions required in the delivery phase as described in Theorem \ref{thm2}. Note that both the results proved in this section are for $M_p=\frac{N}{K}$.
		\subsection{Achievability Scheme}
		Before presenting the general placement and delivery scheme, we give two examples that illustrate the main idea behind the achievability scheme. %The delivery scheme makes two types of transmissions.
		\label{achievability}
		\begin{example}
			\label{example1}
			Consider a CMAP system with a central server having $N=6$ files and $\Lambda=4$ access caches, each capable of storing $M_a=1.5$ files. The access degree for this system is $r=2$. There are $K=\binom{\Lambda}{r}=\binom{4}{2}=6$ users, denoted as $\{1,2\},\{1,3\},\{1,4\},\{2,3\},\{2,4\}$, and, $\{3,4\}$. For this system, we have $t=\frac{4\times1.5}{6}=1$. Each file $W_{n}$ is split into $\binom{\Lambda}{t}=4$ subfiles of equal length as $W_n=\{W_{n,\{1\}},W_{n,\{2\}},W_{n,\{3\}},W_{n,\{4\}}\}, \forall n\in[1,6]$. The contents of the access caches are 
			\begin{align*}
				&Z_1=\{W_{n,\{1\}},\forall n\in[1,6]\},\\
				&Z_2=\{W_{n,\{2\}},\forall n\in[1,6]\},\\
				&Z_3=\{W_{n,\{3\}},\forall n\in[1,6]\},\text{ and},\\
				&Z_4=\{W_{n,\{4\}},\forall n\in[1,6]\}.
			\end{align*}%
			Each access cache stores $\frac{6}{4}=1.5$ files, satisfying its memory constraint. After users connect to access caches, the subfile $W_{n, S}$ is available to those users whose user-index sets have a non-empty intersection with $\mathcal{S}$, that is, $\{\mathcal{U}:\mathcal{S}\cap\mathcal{U} \neq \emptyset\}$. This means that the subfile $W_{n,\mathcal{S}}$ is not available to $\binom{\Lambda-t}{r}=\binom{3}{2}=3$ users. Thus, every subfile is split into $3$ mini-subfiles. The private caches of the users are populated with the mini-subfiles of the subfiles the users do not get when connecting to the access caches. The mini-subfile of the subfile $\mathcal{S}$ of file $n$ that is stored in the private cache of user $\mathcal{U}$ is denoted as $W_{n,\mathcal{S},\mathcal{U}}$. Since each file is divided into $4$ subfiles and each subfile is further divided into $3$ mini-subfiles, the total subpacketization is $F=12$. The contents of the private caches of the users are shown below:
			\begin{align*}
				&Z^p_{\{1,2\}}=\{W_{n,\{3\},\{1,2\}},W_{n,\{4\},\{1,2\}},\forall n\in[1,6]\},\\
				&Z^p_{\{1,3\}}=\{W_{n,\{2\},\{1,3\}},W_{n,\{4\},\{1,3\}},\forall n\in[1,6]\},\\
				&Z^p_{\{1,4\}}=\{W_{n,\{2\},\{1,4\}},W_{n,\{3\},\{1,4\}},\forall n\in[1,6]\},\\
				&Z^p_{\{2,3\}}=\{W_{n,\{1\},\{2,3\}},W_{n,\{4\},\{2,3\}},\forall n\in[1,6]\},\\
				&Z^p_{\{2,4\}}=\{W_{n,\{1\},\{2,4\}},W_{n,\{3\},\{2,4\}},\forall n\in[1,6]\},\text{ and},\\
				&Z^p_{\{3,4\}}=\{W_{n,\{1\},\{3,4\}},W_{n,\{2\},\{3,4\}},\forall n\in[1,6]\}.
			\end{align*}%
			There is $\frac{12}{12}=1$ file in each private cache, satisfying their memory constraint. From now on, the user-index set, the subfile-index set, and the mini-subfile-index set will be compactly written without the set notation. 
			
			We will now explain how the server constructs the transmissions. For a delivery vector $\mathbf{d}=(d_{\mathcal{U}}:\mathcal{U}\subset[1,\Lambda], |\mathcal{U}|=r)$, consider the mini-subfile $W_{d_{12},3,14}$ demanded by user $12$. The server calculates the intersection between the user-index set and the mini-subfile-index set as $I=\{12\cap14\}=\{1\}$. To construct the transmission, the server picks $\{1\}$ from both the user-index set, $12$, and the mini-subfile-index set, $14$, and swaps it with the subfile-index set, $3$, to obtain $W_{d_{23},1,34}$. Next, the server flips the user-index set and the mini-subfile-index set of both these mini-subfiles, obtaining $W_{d_{14},3,12}$ and $W_{d_{34},1,23}$, respectively. Finally, the server performs the XOR operation of these four mini-subfiles, as 
			\begin{align*}
				W_{d_{12},3,14}\oplus W_{d_{23},1,34}\oplus W_{d_{14},3,12}\oplus W_{d_{34},1,23}.
			\end{align*}%
			We show all the transmissions made by the server below:
			\begin{enumerate}
				\item $W_{d_{12},3,14}\oplus W_{d_{23},1,34}\oplus W_{d_{14},3,12}\oplus W_{d_{34},1,23}$
				\item $W_{d_{12},3,24}\oplus W_{d_{13},2,34}\oplus W_{d_{24},3,12}\oplus W_{d_{34},2,13}$
				\item $W_{d_{12},4,13}\oplus W_{d_{24},1,34}\oplus W_{d_{34},1,24}\oplus W_{d_{13},4,12}$
				\item $W_{d_{12},4,23}\oplus W_{d_{14},2,34}\oplus W_{d_{34},2,14}\oplus W_{d_{23},4,12}$
				\item $W_{d_{13},2,14}\oplus W_{d_{23},1,24}\oplus W_{d_{24},1,23}\oplus W_{d_{14},2,13}$, and,
				\item $W_{d_{13},4,23}\oplus W_{d_{14},3,34}\oplus W_{d_{34},3,14}\oplus W_{d_{23},4,13}$.
			\end{enumerate}%
			Since the server makes six transmissions and the subpacketization is $F=12$, the rate is $R=0.5$.
		\end{example}
		In the above example, notice that for every mini-subfile $W_{d_{\mathcal{U}},\mathcal{S},\mathcal{U}^\prime}$, there is always an intersection between the user-index set and the mini-subfile-index set, that is $I=\mathcal{U}\cap\mathcal{U}^\prime\not=\emptyset$. However, for $K\geq 2r+t$, mini-subfiles having $I=\emptyset$ also exist. The following example shows how the server constructs transmission for such mini-subfiles.
		\begin{example}
			Consider $(5,2,2,1,10)-$CMAP coded caching setting. There is a central server with a library of $N=10$ files. The server connects to $K=10$ users, equipped with private caches of capacity $M_p=1$ file. There are $\Lambda=5$ access caches, each capable of storing $M_a=2$ files such that a unique user accesses every $r=2$ caches. For this network, $t=1$. Each file is split into $\binom{\Lambda}{t}=5$ subfiles of equal size as $W_n=\{W_{n,1}, W_{n,2}, W_{n,3}, W_{n,4}, W_{n,5}\}, \forall n\in[1,10]$. The contents of the access caches are:
			\begin{align*}
				&Z_1=\{W_{n,1},\forall n\in[1,10]\},\\
				&Z_2=\{W_{n,2},\forall n\in[1,10]\},\\
				&Z_3=\{W_{n,3},\forall n\in[1,10]\},\\
				&Z_4=\{W_{n,4},\forall n\in[1,10]\},\text{ and}\\
				&Z_5=\{W_{n,5},\forall n\in[1,10]\}.
			\end{align*}
			Every access cache stores $\frac{10}{5}=2$ files, satisfying its memory constraint. Every subfile is further split into $6$ mini-subfiles of equal size. So, we have a subpacketization of $F=30$. The cache contents of the private cache of users are as shown below:
			\begin{align*}
				&Z^p_{12}=\{W_{n,3,12},W_{n,4,12},W_{n,5,12},\forall n\in[1,10]\},\\
				&Z^p_{13}=\{W_{n,2,13},W_{n,4,13},W_{n,5,13},\forall n\in[1,10]\},\\
				&Z^p_{14}=\{W_{n,2,14},W_{n,3,14},W_{n,5,14},\forall n\in[1,10]\},\\
				&Z^p_{15}=\{W_{n,2,15},W_{n,3,15},W_{n,4,15},\forall n\in[1,10]\},\\
				&Z^p_{23}=\{W_{n,1,23},W_{n,4,23},W_{n,5,23},\forall n\in[1,10]\},\\
				&Z^p_{24}=\{W_{n,1,24},W_{n,3,24},W_{n,5,24},\forall n\in[1,10]\},\\
				&Z^p_{25}=\{W_{n,1,25},W_{n,3,25},W_{n,4,25},\forall n\in[1,10]\},\\
				&Z^p_{34}=\{W_{n,1,34},W_{n,2,34},W_{n,5,34},\forall n\in[1,10]\},\\
				&Z^p_{35}=\{W_{n,1,35},W_{n,2,35},W_{n,4,35},\forall n\in[1,10]\},\text{ and},\\
				&Z^p_{45}=\{W_{n,1,45},W_{n,2,45},W_{n,3,45},\forall n\in[1,10]\}.
			\end{align*}%
			Each private cache stores $\frac{30}{30}=1$ file, satisfying its capacity.
			
			We now explain how transmissions are constructed. In this example, there are two types of mini-subfiles: those with $I \neq \emptyset$ and those with $I =\emptyset$. Since Example \ref{example1} illustrates how transmissions are constructed when $I \neq \emptyset$, we focus here on the case where $I = \emptyset$. For a demand vector $\mathbf{d} = (d_{\mathcal{U}} : \mathcal{U} \subseteq [1,\Lambda], |\mathcal{U}| = r)$, consider the mini-subfile $W_{d_{12}, 3, 45}$ requested by user 12. For this mini-subfile, $I = \{12 \cap 45\} = \emptyset$, indicating that there is no overlap between the user-index set, $12$, and the subfile-index set, $45$. The server then selects an element from the user-index set $12$ and swaps it with an element from the subfile-index set $3$, resulting in the creation of mini-subfiles $W_{d_{13}, 2, 45}$ and $W_{d_{23}, 1, 45}$. Subsequently, the server performs an XOR operation on these mini-subfiles to construct the transmission:
			\begin{align*} W_{d_{12},3,45}\oplus W_{d_{13},2,45} \oplus W_{d_{23},1,45}.\end{align*}
			Finally, the server makes the following transmissions for $I\not=\emptyset$:			
			\begin{enumerate}
				\item $W_{d_{12},3,14}\oplus W_{d_{23},1,34}\oplus W_{d_{34},1,23}\oplus W_{d_{14},3,12}$
				\item $W_{d_{12},3,15}\oplus W_{d_{23},1,35}\oplus W_{d_{35},1,23}\oplus W_{d_{15},3,12}$
				\item $W_{d_{12},3,24}\oplus W_{d_{13},2,34}\oplus W_{d_{34},2,13}\oplus W_{d_{24},3,12}$
				\item $W_{d_{12},3,25}\oplus W_{d_{13},2,35}\oplus W_{d_{35},2,13}\oplus W_{d_{25},3,12}$
				\item $W_{d_{12},4,13}\oplus W_{d_{24},1,34}\oplus W_{d_{34},1,24}\oplus W_{d_{13},4,12}$
				\item $W_{d_{12},4,15}\oplus W_{d_{24},1,45}\oplus W_{d_{45},1,24}\oplus W_{d_{15},4,12}$
				\item $W_{d_{12},4,23}\oplus W_{d_{14},2,34}\oplus W_{d_{34},2,14}\oplus W_{d_{23},4,12}$
				\item $W_{d_{12},4,25}\oplus W_{d_{14},2,45}\oplus W_{d_{45},2,14}\oplus W_{d_{25},4,12}$
				\item $W_{d_{12},5,13}\oplus W_{d_{25},1,35}\oplus W_{d_{35},1,25}\oplus W_{d_{13},5,12}$
				\item $W_{d_{12},5,14}\oplus W_{d_{25},1,45}\oplus W_{d_{45},1,25}\oplus W_{d_{14},5,12}$
				\item $W_{d_{12},5,23}\oplus W_{d_{15},2,35}\oplus W_{d_{35},2,15}\oplus W_{d_{23},5,12}$
				\item $W_{d_{12},5,24}\oplus W_{d_{15},2,45}\oplus W_{d_{45},2,15}\oplus W_{d_{24},5,12}$
				\item $W_{d_{13},2,14}\oplus W_{d_{23},1,24}\oplus W_{d_{24},1,23}\oplus W_{d_{14},2,13}$
				\item $W_{d_{13},2,15}\oplus W_{d_{23},1,25}\oplus W_{d_{25},1,23}\oplus W_{d_{15},2,13}$
				\item $W_{d_{13},4,15}\oplus W_{d_{34},1,45}\oplus W_{d_{45},1,34}\oplus W_{d_{15},4,13}$
				\item $W_{d_{13},4,23}\oplus W_{d_{14},3,24}\oplus W_{d_{24},3,14}\oplus W_{d_{23},4,13}$
				\item $W_{d_{13},4,35}\oplus W_{d_{14},3,45}\oplus W_{d_{45},3,14}\oplus W_{d_{35},4,13}$
				\item $W_{d_{13},5,14}\oplus W_{d_{35},1,45}\oplus W_{d_{45},1,35}\oplus W_{d_{14},5,13}$
				\item $W_{d_{13},5,23}\oplus W_{d_{15},3,25}\oplus W_{d_{25},3,15}\oplus W_{d_{23},5,13}$
				\item $W_{d_{13},5,34}\oplus W_{d_{15},3,45}\oplus W_{d_{45},3,15}\oplus W_{d_{34},5,13}$
				\item $W_{d_{14},2,15}\oplus W_{d_{24},1,25}\oplus W_{d_{25},1,24}\oplus W_{d_{15},2,14}$
				\item $W_{d_{14},3,15}\oplus W_{d_{34},1,35}\oplus W_{d_{35},1,34}\oplus W_{d_{15},3,14}$
				\item $W_{d_{14},5,24}\oplus W_{d_{15},4,25}\oplus W_{d_{25},4,15}\oplus W_{d_{24},5,14}$
				\item $W_{d_{14},5,34}\oplus W_{d_{15},4,35}\oplus W_{d_{35},4,15}\oplus W_{d_{34},5,14}$
				\item $W_{d_{23},4,25}\oplus W_{d_{34},2,45}\oplus W_{d_{45},2,34}\oplus W_{d_{25},4,23}$
				\item $W_{d_{23},4,35}\oplus W_{d_{24},3,45}\oplus W_{d_{45},3,24}\oplus W_{d_{35},4,23}$
				\item $W_{d_{23},5,24}\oplus W_{d_{35},2,45}\oplus W_{d_{45},2,35}\oplus W_{d_{24},5,23}$
				\item $W_{d_{23},5,34}\oplus W_{d_{25},3,45}\oplus W_{d_{45},3,25}\oplus W_{d_{34},5,23}$
				\item $W_{d_{24},3,25}\oplus W_{d_{34},2,35}\oplus W_{d_{35},2,34}\oplus W_{d_{25},3,24}$, and,
				\item $W_{d_{24},5,34}\oplus W_{d_{25},4,35}\oplus W_{d_{35},4,25}\oplus W_{d_{34},5,24}$
			\end{enumerate}
			and the following transmissions for $I=\emptyset$:
			\begin{enumerate}
				\item $W_{d_{12},3,45}\oplus W_{d_{13},2,45}\oplus W_{d_{23},1,45}$
				\item $W_{d_{12},4,35}\oplus W_{d_{14},2,35}\oplus W_{d_{24},1,35}$
				\item $W_{d_{12},5,34}\oplus W_{d_{15},2,34}\oplus W_{d_{25},1,34}$
				\item $W_{d_{13},4,25}\oplus W_{d_{14},3,25}\oplus W_{d_{34},1,25}$
				\item $W_{d_{13},5,24}\oplus W_{d_{15},3,24}\oplus W_{d_{35},1,24}$
				\item $W_{d_{14},5,23}\oplus W_{d_{15},4,23}\oplus W_{d_{45},1,23}$
				\item $W_{d_{23},4,15}\oplus W_{d_{24},3,15}\oplus W_{d_{34},2,15}$
				\item $W_{d_{23},5,14}\oplus W_{d_{25},3,14}\oplus W_{d_{35},2,14}$
				\item $W_{d_{24},5,13}\oplus W_{d_{25},4,13}\oplus W_{d_{45},2,13}$, and,
				\item $W_{d_{34},5,12}\oplus W_{d_{35},4,12}\oplus W_{d_{45},3,12}$
			\end{enumerate}
			Server makes two types of transmissions depending on whether a demanded mini-subfile $W_{d_{\mathcal{U}},\mathcal{S},\mathcal{U}^\prime}$ has $I=\emptyset$ or not. Since the server makes $40$ transmissions, the rate $R=\frac{40}{30}=\frac{4}{3}$.
		\end{example}
		We now give the general description of the placement and delivery scheme.

		\textit{Placement Phase}:\label{placementpolicy} First, we describe the placement policy of the access caches. Each file is split into $\binom{\Lambda}{t}$ non-overlapping subfiles of equal size as follows:
		\begin{equation}
			W_n=\{W_{n,\mathcal{T}}:\mathcal{T}\subseteq[1,\Lambda],|\mathcal{T}|=t\},\;\forall n\in[1,N],
		\end{equation}where $t=\frac{\Lambda M_a}{N}$, is the access cache memory replication factor, and the contents of the access cache $a\in[1,\Lambda]$ are given as:
		\begin{equation}
			\label{placementaccess}
			Z_a=\{W_{n,\mathcal{T}}:a\in\mathcal{T},\mathcal{T}\subseteq[1,\Lambda],|\mathcal{T}|=t,\forall n\in[1,N]\}.
		\end{equation}
		
		Note that each access cache is populated by $\frac{N\binom{\Lambda-1}{t-1}}{\binom{\Lambda}{t}}=M_a$ files, satisfying the memory constraint. 
		
		Now, we describe the placement strategy of the private caches. For a user $\mathcal{U}$, the server populates its private cache with parts of the subfiles $\mathcal{U}$ does not obtain from the access caches it connects to. Each subfile is wanted by $\binom{\Lambda-t}{r}$ users and hence, is further divided into $\binom{\Lambda-t}{r}$ mini-subfiles. The mini-subfile of subfile $\mathcal{T}$ of the file $n$, stored in the private cache of $\mathcal{U}$, is denoted as $W_{n,\mathcal{T},\mathcal{U}}$. The contents of the private cache of user $\mathcal{U}$ are given as:
		\begin{align}
			\label{placementprivatecaches}
			Z^p_\mathcal{U}=\{W_{n,\mathcal{T},\mathcal{U}}:\mathcal{T}\subseteq[1,\Lambda]\setminus\mathcal{U},\;|\mathcal{T}|=t,\forall n\in[1,N]\}.
		\end{align}
		Each private cache stores $\frac{N\binom{\Lambda-r}{t}}{\binom{\Lambda-t}{r}\binom{\Lambda}{t}}=\frac{N}{\binom{\Lambda}{r}}=M_p$ files, satisfying the memory constraint. Under the outlined placement policies, there is no overlap in the contents of a user's private cache and the access caches it connects to.
		\begin{algorithm*}[t]
			\caption{Algorithm for generating transmission during delivery phase}
			\label{Algo1}
			\hspace*{\algorithmicindent} \textbf{Input:} $\mathbf{d}=(d_{\mathcal{U}}:\mathcal{U}\subseteq[1,\Lambda],|\mathcal{U}|=r)$, $\mathcal{Z}=(\mathcal{Z}_{\mathcal{U}}:\mathcal{U}\subseteq[1,\Lambda],|\mathcal{U}|=r)$. \\
			\hspace*{\algorithmicindent} \textbf{Output:} The set of transmissions $T$.
			\begin{algorithmic}[1]
				\State \label{line1}Initialize $T=\emptyset$.
				\State \label{line2}For each user $\mathcal{U}$, define the user-demand set $\mathcal{D}_\mathcal{U}=\{(\mathcal{S},\mathcal{U}^\prime):\mathcal{S}\subseteq[1,\Lambda]\setminus\mathcal{U},|\mathcal{S}|=t,\;\mathcal{U}^\prime\subseteq[1,\Lambda],|\mathcal{U}^\prime|=r,\mathcal{U}^\prime\not=\mathcal{U},\mathcal{S}\cap\mathcal{U}^\prime=\emptyset\}$.
				\For{$\mathcal{U}\subseteq[1,\Lambda],|\mathcal{U}|=r$}\label{line3}
				\While{$\mathcal{D}_\mathcal{U}\not=\emptyset$}\label{line4}
				\State \label{line5}Select an element $(\mathcal{S},\mathcal{U}^\prime)$ from $\mathcal{D}_\mathcal{U}$.
				\State \label{line6}For $\mathcal{S},\mathcal{U}^\prime$, define ${I}=\mathcal{U}\cap\mathcal{U}^\prime$.
				\If{$|I|>0$}\label{line7}
				\State\label{line8} $T^\prime=flip\left(W_{d_{\mathcal{U}},\mathcal{S},\mathcal{U}^\prime}\oplus\bigoplus\limits_{i=1}^{min(|I|,t)} swap_o(W_{d_{\mathcal{U}},\mathcal{S},\mathcal{U}^\prime},i)\right).$
				\Else\label{line9}
				\State\label{line10} $T^\prime=W_{d_{\mathcal{U}},\mathcal{S},\mathcal{U}^\prime}\oplus\bigoplus\limits_{i=1}^{min(r,t)} swap_{no}(W_{d_{\mathcal{U}},\mathcal{S},\mathcal{U}^\prime},i).$
				\EndIf\label{line11}
				\State\label{line12} $T\gets T\cup T^\prime$.
				\State \label{line13}Let $S_c=\{(\mathcal{U},\mathcal{S},\mathcal{U}^\prime):W_{d_{\mathcal{U}},\mathcal{S},\mathcal{U}^\prime}\text{ is a mini-subfile in $T^\prime$}\}$. For each $(\hat{\mathcal{U}},\hat{\mathcal{S}},\hat{\mathcal{U}^\prime})\in S_c$, do $\mathcal{D}_{\hat{\mathcal{U}}}\gets\mathcal{D}_{\hat{\mathcal{U}}}\setminus (\hat{\mathcal{S}},\hat{\mathcal{U}^\prime})$
				\EndWhile\label{line14}
				\EndFor\label{line15}
			\end{algorithmic}
			\noindent\rule{\textwidth}{1pt}
		\end{algorithm*}
		Before we explain the delivery phase, we define three functions, namely, $flip$, $swap_o$, and $swap_{no}$.
		\begin{defn}
			For a subfile $W_{d_{\mathcal{U}},\mathcal{S},\mathcal{U}^\prime}$, the function \textit{flip} is defined as $flip(W_{d_{\mathcal{U}},\mathcal{S},\mathcal{U}^\prime})=W_{d_{\mathcal{U}},\mathcal{S},\mathcal{U}^\prime}\oplus W_{d_{\mathcal{U}^\prime},\mathcal{S},\mathcal{U}}$.	
		\end{defn}
		\begin{remark}
			$flip(\bigoplus\limits_{i=1}^{n}W_{d_{\mathcal{U}_i},\mathcal{S}_i,\mathcal{U}^\prime_i})=\bigoplus\limits_{i=1}^{n}(flip(W_{d_{\mathcal{U}_i},\mathcal{S}_i,\mathcal{U}^\prime_i}))$.
		\end{remark}
		\begin{example}
			For a mini-subfile $W_{d_{12},3,14}$, we have $flip(W_{d_{12},3,14})=W_{d_{12},3,14}\oplus W_{d_{14},3,12}$.
		\end{example}
		\begin{defn}
			For a subfile $W_{d_{\mathcal{U}},\mathcal{S},\mathcal{U}^\prime}$, such that $\mathcal{U}$ and $\mathcal{U}^\prime$ overlap, i.e., $I=\mathcal{U}\cap\mathcal{U}^\prime\not=\emptyset$, the function $swap_o$ is defined as 			
			\begin{equation*}
				swap_o(W_{d_{\mathcal{U}},\mathcal{S},\mathcal{U}^\prime},i)=\bigoplus\limits_{\substack{\widetilde{\mathcal{U}}\subseteq I,\\|\widetilde{\mathcal{U}}|=i,\\\widetilde{\mathcal{S}}\subseteq\mathcal{S},\\|\widetilde{\mathcal{S}}|=i}} W_{d_{\{\mathcal{U}\cup\widetilde{\mathcal{S}}\}\setminus\widetilde{\mathcal{U}}},\{\mathcal{S}\cup\widetilde{\mathcal{U}}\}\setminus\widetilde{\mathcal{S}},\{\mathcal{U}^\prime\cup\widetilde{\mathcal{S}}\}\setminus\widetilde{\mathcal{U}}}.
			\end{equation*}
		\end{defn}
		\begin{example}
			Consider a mini-subfile $W_{d_{123},45,126}$ for which $swap_o(W_{d_{123},45,126},1)=W_{d_{234},15,246}\oplus W_{d_{134},25,146}\oplus W_{d_{235},14,256}\oplus W_{d_{135},24,156}$. Observe that all possible $1-$subsets of the intersection set $\mathcal{U} \cap \mathcal{U}^\prime$ have been swapped with all possible $1-$subsets of the subfile-index set.
			\end{example}
		For the mini-subfile $W_{d_{\mathcal{U}},\mathcal{S},\mathcal{U}^\prime}$, such that there is no overlap between $\mathcal{U}$ and $\mathcal{U}^\prime$, we define the function $swap_{no}$ as follows. 
		\begin{defn}
			For a subfile $W_{d_{\mathcal{U}},\mathcal{S},\mathcal{U}^\prime}$, such that $\mathcal{U}\cap\mathcal{U}^\prime=\emptyset$, the function $swap_{no}$ is defined as 
			\begin{equation*}
				swap_{no}(W_{d_{\mathcal{U}},\mathcal{S},\mathcal{U}^\prime},i)=\bigoplus\limits_{\substack{\widetilde{\mathcal{U}}\subseteq\mathcal{U},\\|\widetilde{\mathcal{U}}|=i,\\\widetilde{\mathcal{S}}\subseteq\mathcal{S},\\|\widetilde{\mathcal{S}}|=i}} W_{d_{\{\mathcal{U}\cup\widetilde{\mathcal{S}}\}\setminus\widetilde{\mathcal{U}}},\{\mathcal{S}\cup\widetilde{\mathcal{U}}\}\setminus\widetilde{\mathcal{S}},\mathcal{U}^\prime}.
			\end{equation*}					
			
		\end{defn}
		\begin{example}
			For a mini-subfile $W_{d_{123},45,678}$, we have $swap_{no}(W_{d_{123},45,678},2)=W_{d_{345},12,678}+W_{d_{245},13,678}+W_{d_{145},23,678}$, which is obtained by swapping every $2-$subset of the user-index set $123$ with the sub-file index set $45$.
		\end{example}
		\textit{Delivery Phase}: For the demand vector $\mathbf{d}$, the server broadcasts the set of transmissions $T$ returned by Algorithm \ref{Algo1}. We will now explain the working of Algorithm \ref{Algo1}.
		
		Given a demand vector $\mathbf{d}$ and the placement policy described in section \ref{achievability}, Algorithm \ref{Algo1} returns the set of transmissions $T$ required to satisfy the demands of all the $K$ users. For each user $\mathcal{U}$, the algorithm first defines the user-demand set $\mathcal{D}_\mathcal{U}$, which contains the indices of all the mini-subfiles that are wanted by that user. For instance, in Example \ref{example1}, the user-demand sets for the users are as given below:
		\begin{align*}
			\mathcal{D}_{12}=\{(3,14),(3,24),(4,13),(4,23)\},\\
			\mathcal{D}_{13}=\{(2,14),(2,34),(4,12),(4,23)\},\\
			\mathcal{D}_{14}=\{(2,13),(2,34),(3,12),(3,24)\},\\
			\mathcal{D}_{23}=\{(1,24),(1,34),(4,12),(4,13)\},\\
			\mathcal{D}_{24}=\{(1,23),(1,34),(3,12),(3,14)\},\\
			\mathcal{D}_{34}=\{(1,23),(1,24),(2,13),(2,14)\}.
		\end{align*} The algorithm then selects a user and picks an element from the user-demand set of this user. For this element, the algorithm calculates the intersection between the user-index set of the user and the mini-subfile-index set of the selected mini-subfile. Let us say, for Example \ref{Algo1}, the algorithm picks the user $12$ and selects the element $(3,14)$, describing the mini-subfile $W_{d_{12},3,14}$. The intersection for the mini-subfile $W_{d_{12},3,14}$ is $I=\{12\cup 14\}=\{1\}$. Depending on whether this intersection is empty or not, the algorithm constructs the transmissions described in Line \ref{line8} or Line \ref{line10}, respectively. Since $I\not=\emptyset$ for the mini-subfile $W_{d_{12},3,14}$ being considered, the algorithm constructs the transmission $W_{d_{12},3,14}\oplus W_{d_{23},1,34}\oplus W_{d_{14},3,12}\oplus W_{d_{34},1,23}$ as described in Line \ref{line8}. Finally, the algorithm removes the indices of all the mini-subfiles in the constructed transmission from their respective user-demand sets. Hence, the user-demand sets of the users in Example \ref{example1} after construction of the above transmission are:
		\begin{align*}
			&\mathcal{D}_{12}=\{(3,24),(4,13),(4,23)\},\\
			&\mathcal{D}_{13}=\{(2,14),(2,34),(4,12),(4,23)\},\\
			&\mathcal{D}_{14}=\{(2,13),(2,34),(3,24)\},\\
			&\mathcal{D}_{23}=\{(1,24),(4,12),(4,13)\},\\
			&\mathcal{D}_{24}=\{(1,23),(1,34),(3,12),(3,14)\},\\
			&\mathcal{D}_{34}=\{(1,24),(2,13),(2,14)\}.
		\end{align*}

		\textit{Decodability}: Algorithm \ref{Algo1} generates two types of transmissions based on whether $I=\mathcal{U}\cap\mathcal{U}^\prime=\emptyset$ or $I=\mathcal{U}\cap\mathcal{U}^\prime\not=\emptyset$. Consider a transmission for the no overlap case $|I|=0$, $W_{d_{\mathcal{U}},\mathcal{S},\mathcal{U}^\prime}\oplus\bigoplus\limits_{i=1}^{min(r,t)} swap_{no}(W_{d_{\mathcal{U}},\mathcal{S},\mathcal{U}^\prime},i)$. Every mini-subfile in the function $swap_{no}(W_{d_{\mathcal{U}},\mathcal{S},\mathcal{U}^\prime},i)$ has non-zero intersection between its subfile-index set and $\mathcal{U}$. Hence, $\mathcal{U}$ will have these mini-subfiles from the access caches it connects to and can obtain its desired mini-subfile. Now consider a transmission where $|I|>0$, $flip(W_{d_{\mathcal{U}},\mathcal{S},\mathcal{U}^\prime}\oplus\bigoplus\limits_{i=1}^{min(|I|,t)} swap_o(W_{d_{\mathcal{U},\mathcal{S},\mathcal{U}^\prime}},i))$. Note that the user $\mathcal{U}$ can remove all mini-subfiles from this transmission, except $W_{d_{\mathcal{U}},\mathcal{S},\mathcal{U}^\prime}$ and $W_{d_{\mathcal{U}^\prime},\mathcal{S},\mathcal{U}}$ using the contents of its access caches. But the mini-subfile $W_{d_{\mathcal{U}^\prime},\mathcal{S},\mathcal{U}}$ is in its private cache. Thus, the user $\mathcal{U}$ can decode its desired mini-subfile. Since both types of transmissions are decodable and Algorithm \ref{Algo1} runs until all the user-demand sets are empty, every user is able to obtain all the mini-subfiles of its desired file.
		
		\textit{Performance of Algorithm \ref{Algo1}}: For the $|I|=0$ case, each transmission has $\sum\limits_{j=0}^{min(r,t)} \binom{r}{j}\binom{t}{j}$ mini-subfiles. Using Vandermonde's identity, we know that $\sum\limits_{j=0}^{min(r,t)} \binom{r}{j}\binom{t}{j}=\binom{t+r}{t}$. For the $|I|=i$ case, each transmission has $2\sum\limits_{j=1}^{min(i,t)} \binom{t}{j}\binom{i}{j}=2\binom{t+i}{i}$ mini-subfiles.
		Since there are $\binom{\Lambda}{t}\binom{\Lambda-t}{r}\binom{\Lambda-r-t}{r}$ mini-subfiles that have $|I|=0$ and $\binom{\Lambda}{t}\binom{\Lambda-t}{r}\binom{r}{i}\binom{\Lambda-r-t}{r-i}$ mini-subfiles that have $|I|=i$, and each file is divided into $\binom{\Lambda}{t}\binom{\Lambda-t}{r}$ mini-subfiles, the rate is
		
		\begin{align}
			R=&\frac{\binom{\Lambda}{t}\binom{\Lambda-t}{r}\binom{\Lambda-r-t}{r}}{\binom{\Lambda}{t}\binom{\Lambda-t}{r}\binom{t+r}{r}} + \sum\limits_{i=1}^{r-1} \frac{\binom{\Lambda}{t}\binom{\Lambda-t}{r}\binom{r}{i}\binom{\Lambda-r-t}{r-i}}{2\binom{\Lambda}{t}\binom{\Lambda-t}{r}\binom{t+i}{i}}\nonumber\\
			=&\frac{\binom{\Lambda-r-t}{r}}{\binom{t+r}{t}}+\sum\limits_{i=1}^{r-1}\frac{\binom{r}{i}\binom{\Lambda-t-r}{r-i}}{2\binom{t+i}{i}}.
		\end{align}
		\begin{remark}
			\label{remark3}
			It can be seen that for the case where $|I|=0$, the coding gain, defined as the total number of users benefiting from
			each transmission, is $\binom{t+r}{r}$ and when $|I|=i$, the coding gain is $2\binom{t+i}{i}$. Hence, as the access cache memory replication factor $t$ increases, the coding gain for both types of transmissions increases, while as the access degree $r$ increases, the coding gain of the transmissions of $|I|=0$ case increases. %It is reflected in Fig. \ref{fig2}, where the rate decreases as $t$ increases for a given $r$ or as $r$ increases for a given $t$. 
		\end{remark}
		We will now explain how memory sharing is done for the CMAP coded caching system.
		\begin{remark}
			Consider $M_a$ such that $t=\frac{\Lambda M_a}{N}$ is not an integer. Let $M_1=\frac{\lceil t\rceil N}{\Lambda}$ and $M_2=\frac{\lfloor t\rfloor N}{\Lambda}$. Since $M_a=\frac{tN}{\Lambda}$, we know that $M_2\leq M_a\leq M_1$. Hence, $M_a$ can be written as 
			\begin{align*}
				M_a=\alpha M_1 + (1-\alpha)M_2,
			\end{align*}for some $0\leq\alpha\leq 1$. The file $W_{n}$ is split into $W_{n}^\alpha$, of $\alpha B$ bits, and $W_{n}^{(1-\alpha)}$, of $(1-\alpha)B$ bits, respectively, $\forall n\in[1,N]$. The file $W_{n}^\alpha$ is further broken down into subfiles as $W_{n}^\alpha=\{W_{n,\mathcal{S}}^\alpha: \mathcal{S}\subseteq[1,\Lambda],|\mathcal{S}|=\lceil t\rceil\}$, while the file $W_{n}^{(1-\alpha)}$ is broken into subfiles as $W_{n}^{(1-\alpha)}=\{W_{n,\mathcal{S}}^{(1-\alpha)}: \mathcal{S}\subseteq[1,\Lambda],|\mathcal{S}|=\lfloor t\rfloor\}$. The access caches are filled with subfiles $W_{n,\mathcal{S}}^{\alpha}$ as described in \eqref{placementaccess}, for $t=\lceil t\rceil$ and with subfiles $W_{n,\mathcal{S}}^{(1-\alpha)}$, as described in \eqref{placementaccess}, for $t=\lfloor t\rfloor$. Thus, every access cache stores $N\alpha\binom{\Lambda-1}{\lceil t\rceil-1}B+N(1-\alpha)\binom{\Lambda-1}{\lfloor t\rfloor-1}B$ bits, which is equivalent to $N\alpha\frac{\binom{\Lambda-1}{\lceil t\rceil-1}}{\binom{\Lambda}{\lceil t\rceil}}+N(1-\alpha)\frac{\binom{\Lambda-1}{\lfloor t\rfloor-1}}{\binom{\Lambda}{\lfloor t\rfloor}}=\frac{N\alpha\lceil t\rceil}{\Lambda}+\frac{N(1-\alpha)\lfloor t\rfloor}{\Lambda}=\alpha M_1 +(1-\alpha)M_2=M_a$ files, satisfying its memory constraint. 
			
			The private caches of the users will be populated with the mini-subfiles of $W_{n,\mathcal{S}}^{\alpha}$ and $W_{n,\mathcal{S}}^{(1-\alpha)}$ as described in \eqref{placementprivatecaches} for $\lceil t\rceil$ and $\lfloor t \rfloor$. Every private cache stores $\frac{\alpha N\binom{\Lambda-r}{\lceil t\rceil}}{\binom{\Lambda-\lceil t\rceil}{r}\binom{\Lambda}{\lceil t\rceil}}+\frac{(1-\alpha)N\binom{\Lambda-r}{\lfloor t\rfloor }}{\binom{\Lambda-\lfloor t\rfloor}{r}\binom{\Lambda}{\lfloor t\rfloor}}=\frac{\alpha N}{\binom{\Lambda}{r}}+\frac{(1-\alpha) N}{\binom{\Lambda}{r}}=\alpha M_p+(1-\alpha)M_p=M_p$, satisfying its memory constraint. The rate corresponding to $t=\lceil t\rceil$ is $\alpha R_1$ and the rate corresponding to $t=\lfloor t\rfloor$ is $(1-\alpha)R_2$, respectively. Thus,% We will now explain how the rate at memory $M_a$ is calculated. We denote $R_{M_1}$ as the rate at memory point $M_1$, $R_{M_2}$ as the rate at memory point $M_2$, and, $R_{M_a}$ at memory point $M_a$. Rate $R_{M_a}$ is calculated as:
			\begin{align*}
				R_{M_a}=\alpha R_{M_1}+(1-\alpha) R_{M_2}.
			\end{align*} 
		\end{remark}
		%		After explaining memory sharing, we will discuss how the proposed scheme relates to the scheme presented in \cite{PD}.
		\begin{remark}
			Consider a $(\Lambda,r, M_a, M_p=0, N)-$CMAP coded caching system, where private caches have no memory. Users solely rely on the cache contents of the access caches they connect to. This scenario mirrors the settings explored in the MAN scheme for CMACC network\cite{PD}. In this CMAP coded caching system, mini-subfiles follow the structure $W_{d_{\mathcal{U}},\mathcal{S},\emptyset}$, since $M_p=0$. With $M_p=0$, the cache contents user $\mathcal{U}$ has access to is $\mathcal{Z}_{\mathcal{U}}=\{W_{d_{\mathcal{U}},\mathcal{S},\emptyset}:\mathcal{S}\subseteq[1,\Lambda],|\mathcal{S}|=t,\mathcal{U}\cap\mathcal{S}\not=\emptyset\}$. Consequently, the subpacketization for this CMAP system equals $F=\binom{\Lambda}{t}$, which is equal to the subpacketization of the MAN scheme for CMACC network\cite{PD}, for $t=t$. For the mini-subfile $W_{d_\mathcal{U},\mathcal{S},\emptyset}$, we have $I=\mathcal{U}\cap\emptyset=\emptyset$, that is, $|I|=0$. Since $|I|=0$ for every mini-subfile, the transmission made for the mini-subfile $W_{d_{\mathcal{U}},\mathcal{S},\emptyset}$ will be of the form $W_{d_{\mathcal{U}},\mathcal{S},\emptyset}\oplus\bigoplus\limits_{i=1}^{min(r,t)} swap_{no}(W_{d_{\mathcal{U}},\mathcal{S},\emptyset},i)$. Hence, every transmission will be a coded combination of $\binom{t+r}{r}$ mini-subfile and a transmission will be made for every $\mathcal{U}$ and $\mathcal{S}$ such that $\mathcal{U}\cap\mathcal{S}=\emptyset$. Therefore, a transmission is made for every subset of the set $[1,\Lambda]$ of cardinality $t+r$. Thus, we get a rate $R=\frac{\binom{\Lambda}{t+r}}{\binom{\Lambda}{t}}$. This is the same delivery scheme as the MAN scheme for the CMACC network\cite{PD}.
		\end{remark}
		Hence, the proposed scheme specializes to the scheme present in \cite{PD}, when private caches have no memory.
		\subsection{Alpha Bound}
		\label{alphabound}
		The proof of Theorem \ref{thm2} formulates the delivery phase as an ICP $\mathcal{I}$ as was done in \cite{KTR}. We find a lower bound on $\alpha(\mathcal{I})$ and use it to lower bound the number of transmissions made in the delivery phase. We define $\mathcal{U}_i$ as the $i^{\text{th}}\text{ user},i\in\left[1,\binom{\Lambda}{r}\right]$, and $\mathcal{U}^\mathcal{S}_j$ as the $j^{\text{th}}\text{ user},j\in\left[1,\binom{\Lambda-t}{r}\right]$, who wants subfile $\mathcal{S}$, respectively, when the users are arranged lexicographically. We construct the set $B(\mathbf{d})=B_1(\mathbf{d})\cup B_2(\mathbf{d})$, whose elements are messages of the ICP $\mathcal{I}$ such that the set of indices of the messages in $B(\mathbf{d})$ forms a generalized independent set, where,
		\begin{equation*}
			\begin{split}
				B_1(\mathbf{d})\text{=}\bigcup\limits_{i=1}^{\Lambda-r-t+1}\bigcup\limits_{k=i+1}^{\binom{\Lambda-t}{r}} \Big\{W_{d_{\mathcal{U}_i},\mathcal{S},\mathcal{U}^{\mathcal{S}}_k}:\mathcal{S}\subseteq[r+i,\Lambda],|S|=t\Big\}
			\end{split}
		\end{equation*}
		
		\begin{equation*}\text{\normalsize and, }
			\begin{split}
				B_2(\mathbf{d})=\bigcup\limits_{m=m^\prime}^{\binom{\Lambda-t}{r}}\bigcup\limits_{k=m+1}^{\binom{\Lambda-t}{r}}\Big\{W_{d_{\mathcal{U}^{\mathcal{S}^\prime}_m},\mathcal{S}^\prime,U^{\mathcal{S}^\prime}_k}\Big\},
			\end{split}
		\end{equation*}
		where $\mathcal{S}^\prime=[\Lambda-t+1,\Lambda]$ and $m^\prime=\Lambda-r-t+2$. Let $H(\mathbf{d})$ be the set of indices of the messages in $B(\mathbf{d})$.
		
		%		\begin{claim}
			\textit{Claim:} $H(\mathbf{d})$ forms a generalized independent set.
			%		\end{claim}
		%		\begin{proof}
			
			Each message in $B(\mathbf{d})$ is demanded by one receiver. Hence, all the subsets of $H(\mathbf{d})$ of size one are present in $\mathcal{J}(\mathcal{I})$. Consider any set $C=\{W_{d_{\mathcal{U}_{i_1}},\mathcal{S}_{j_1},\mathcal{U}^\prime_{{l_1}}},W_{d_{\mathcal{U}_{i_2}},\mathcal{S}_{{j_2}},\mathcal{U}^\prime_{{l_2}}},\cdots,W_{d_{\mathcal{U}_{i_c}},\mathcal{S}_{{j_c}},\mathcal{U}^\prime_{{l_c}}},\} \subseteq B(\mathbf{d}),$ where, $i_1 \leq i_2 \leq\cdots\leq i_k$. Consider the message $W_{d_{\mathcal{U}_{i_1}},\mathcal{S}_{j_1},\mathcal{U}^\prime_{{l_1}}}$. The receiver demanding this message has no other message in $C$ as side information. Thus indices of messages in $C$ lie in $\mathcal{J}(\mathcal{I})$ and any subset of $H(\mathbf{d})$ will lie in $\mathcal{J}(\mathcal{I})$.
			%		\end{proof}
		
		As $H(\mathbf{d})$ is a generalized independent set, we have $\alpha\geq|H(\mathbf{d})|$ as $|H(\mathbf{d})|$ is equal to $|B(\mathbf{d})|$. Since both the terms in $B(\mathbf{d})$ are disjoint, we count the elements in $B_1(\mathbf{d})$ and $B_2(\mathbf{d})$. Consider a user $\mathcal{U}_i$ in $B_1(\mathbf{d})$. The number of subfiles corresponding to $\mathcal{U}_i$ is $\binom{\Lambda-r-i+1}{t}$ and the number of mini-subfiles corresponding to these subfiles is $\binom{\Lambda-r}{t}-i$. Thus, $|B_1(\mathbf{d})|=\underbrace{\sum\limits_{i=0}^{\Lambda-r-t} \binom{\Lambda-r-i}{t}\left[\binom{\Lambda-t}{r}-1\right]}_{\text{Term $|B_{1,1}(\mathbf{d})|$}}-\underbrace{\sum\limits_{i=0}^{\Lambda-r-t} i\binom{\Lambda-r-i}{t}}_{\text{Term $|B_{1,2}(\mathbf{d})|$}}.$
		Using Hockey-Stick identity, we get	 
		\begin{equation*}
			|B_{1,1}(\mathbf{d})|=\left[\binom{\Lambda-t}{r}-1\right]\binom{\Lambda-r+1}{t+1}\text{\normalsize\;     and, } 
		\end{equation*}
		\begin{equation*}
			|B_{1,2}(\mathbf{d})|=(\Lambda-r+1)\binom{\Lambda-r+1}{t+1}-(t+1)\binom{\Lambda-r+2}{t+2},
		\end{equation*}
		which, upon further simplification, leads to
		\begin{equation*}
			\begin{split}
				|B_1(\mathbf{d})|=&\binom{\Lambda-t}{r}\binom{\Lambda-r+1}{t+1}-\binom{\Lambda-r+2}{t+2}.
			\end{split}
		\end{equation*}
		We now consider $B_2(\mathbf{d})$. For a user $\mathcal{U}^{\mathcal{S}^\prime}_m$,  there are $\binom{\Lambda-t}{r}-m$ mini-subfiles of the subfile $\mathcal{S}^\prime$ in $B_2(\mathbf{d})$ which implies 
		\begin{equation*}
			\begin{split}
				&|B_2(\mathbf{d})|=\sum\limits_{m=\Lambda-r-t+2}^{\binom{\Lambda-t}{r}}\left[\binom{\Lambda-t}{r}-m\right].
			\end{split}
		\end{equation*}
		It can be observed that $|B_2(\mathbf{d})|=1+2+\cdots+\binom{\Lambda-t}{r}-\Lambda+r+t-2$. Hence,
		\begin{equation*}
			|B_2(\mathbf{d})|=\frac{\left(\binom{\Lambda-t}{r}-\Lambda+r+t-2\right)\left(\binom{\Lambda-t}{r}-\Lambda+r+t-1\right)}{2}.
		\end{equation*}
		Finally, we have, 
		\begin{equation*}
			\begin{split}
				\alpha\geq&\binom{\Lambda-t}{r}\binom{\Lambda-r+1}{t+1}-\binom{\Lambda-r+2}{t+2}+\\&\frac{\left(\binom{\Lambda-t}{r}-\Lambda+r+t-2\right)\left(\binom{\Lambda-t}{r}-\Lambda+r+t-1\right)}{2}.
			\end{split}
		\end{equation*}
		
		%	 The theorem is proved since $\alpha$ lower bounds $T$.
		%		
		%			\begin{equation*}
			%				\begin{split}
				%					|B(\mathbf{d})|=&\binom{\Lambda-t}{r}\binom{\Lambda-r+1}{t+1}-\binom{\Lambda-r+2}{t+2}+\\&\frac{\left(\binom{\Lambda-t}{r}-\Lambda+r+t-2\right)\left(\binom{\Lambda-t}{r}-\Lambda+r+t-1\right)}{2}.
				%				\end{split}
			%			\end{equation*}
		%		
		Since $\alpha$ is the cardinality of the maximal generalized independent set, we have $\alpha\geq |H(\mathbf{d})|=|B(\mathbf{d})|$.
		The theorem is proved since $\alpha$ lower bounds $T$.
		\section{Numerical Comparison}
		\label{numericalcomparison}
		In this section, we compare the rate of the proposed scheme in Theorem \ref{thm2} with the upper and lower bounds in Proposition \ref{prop1}, the index-coding based lower bound in Theorem \ref{thm3}, normalized by the subpacketization of the proposed scheme, and the cut-set bound derived in Theorem \ref{thm1}. We provide numerical plots of the rate $R$ for different values of the access degree $r$ and the access cache memory replication factor, $t$ for a system with $\Lambda=6$ access caches, $N$ files, and $K = \binom{\Lambda}{r}$ users such that $N=K$, and $M_p= \frac{N}{K}=1$. The three sets of plots in Fig. \ref{fig2} correspond to $r=2,3$, and $r=4$ cases with $t$ taking values in $[1,\Lambda]$. It can be observed that the rate $R$ approaches $R^{\textasteriskcentered}_D(rM_a+M_p)$ as either $r$ or $t$ increases.
		
		We provide Fig. \ref{fig3}, Fig. \ref{fig4}, Fig \ref{fig5}, and, Fig. \ref{fig6} to further illustrate Remark \ref{remark3}. Fig. \ref{fig3} and Fig. \ref{fig4} correspond to a CMAP system with $\Lambda=6$ access caches, $N$ files, and $K$ users such that each user connects to $r=2$ and $r=3$ access caches respectively. For the access cache memory replication factor $t\in[1,\Lambda]$, the rate of the achievable scheme $R$, rate of the MAN scheme\cite{MAN} such that each cache has a memory of $rM_a+M_p$, rate of the MAN scheme for CMACC network\cite{PD} such that each cache has a capacity of $M_a+\frac{M_p}{r}$, the lower bound on the optimal worst-case rate derived in Theorem \ref{thm1}, and, the lower bound in Theorem \ref{thm3}, normalized by the subpacketization have been plotted. It can be seen from Fig. \ref{fig3} and Fig. \ref{fig4} that the rate of the proposed scheme $R$ moves closer to the lower bound in Theorem \ref{thm3} as $t$ increases for a fixed $r$.
		
		Similarly, Fig. \ref{fig5} and Fig. \ref{fig6} correspond to a CMAP coded caching system with $\Lambda=6$ access caches with a capacity of $M_a=\frac{N}{6}$ and $M_a=\frac{N}{3}$, respectively. A central server with $N$ files and $K$ users connects to the system, with the access degree $r\in[2,\Lambda]$. For this system, the rate of the achievable scheme $R$, the rate of the MAN scheme\cite{MAN}, and the rate of the MAN scheme for CMACC network\cite{PD} keeping the total memory accessed by a user the same in all the three cases, as well as the lower bound in Theorem \ref{thm3}, and, the lower bound in Theorem \ref{thm3}, normalized by the subpacketization have been plotted. It can be seen from Fig. \ref{fig5} and Fig. \ref{fig6} that the rate of the proposed scheme $R$ moves closer to the lower bound in Theorem \ref{thm3} as $r$ increases for a fixed $t$.
		\begin{figure}
			\includegraphics[width=0.8\textwidth,height=0.43\textwidth]{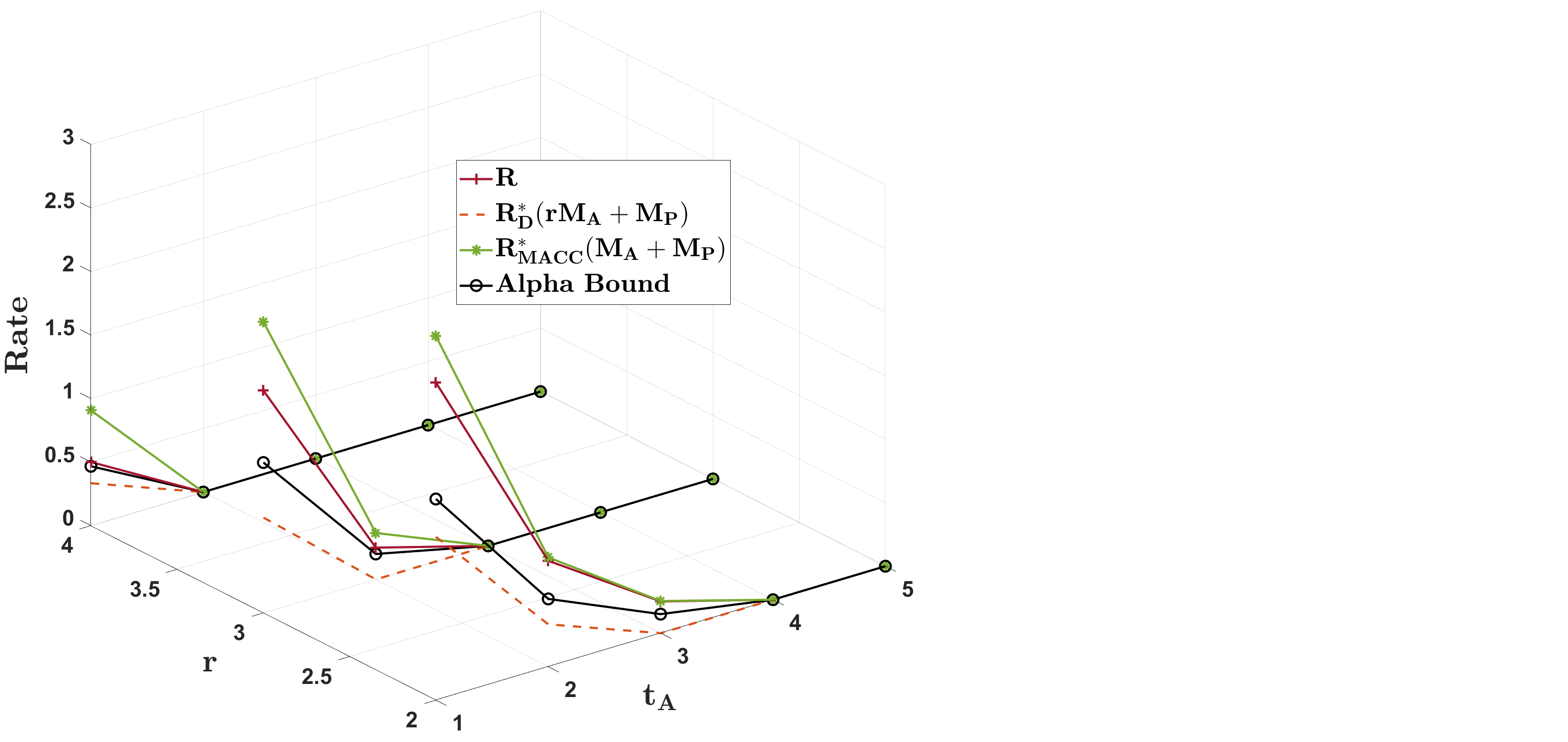}
			\caption{Rate vs. $r$ and $t$ for $M_p = \frac{N}{K}$.}
			\label{fig2}
		\end{figure}
		\begin{figure}
			\includegraphics[width=\textwidth,height=0.5\textwidth]{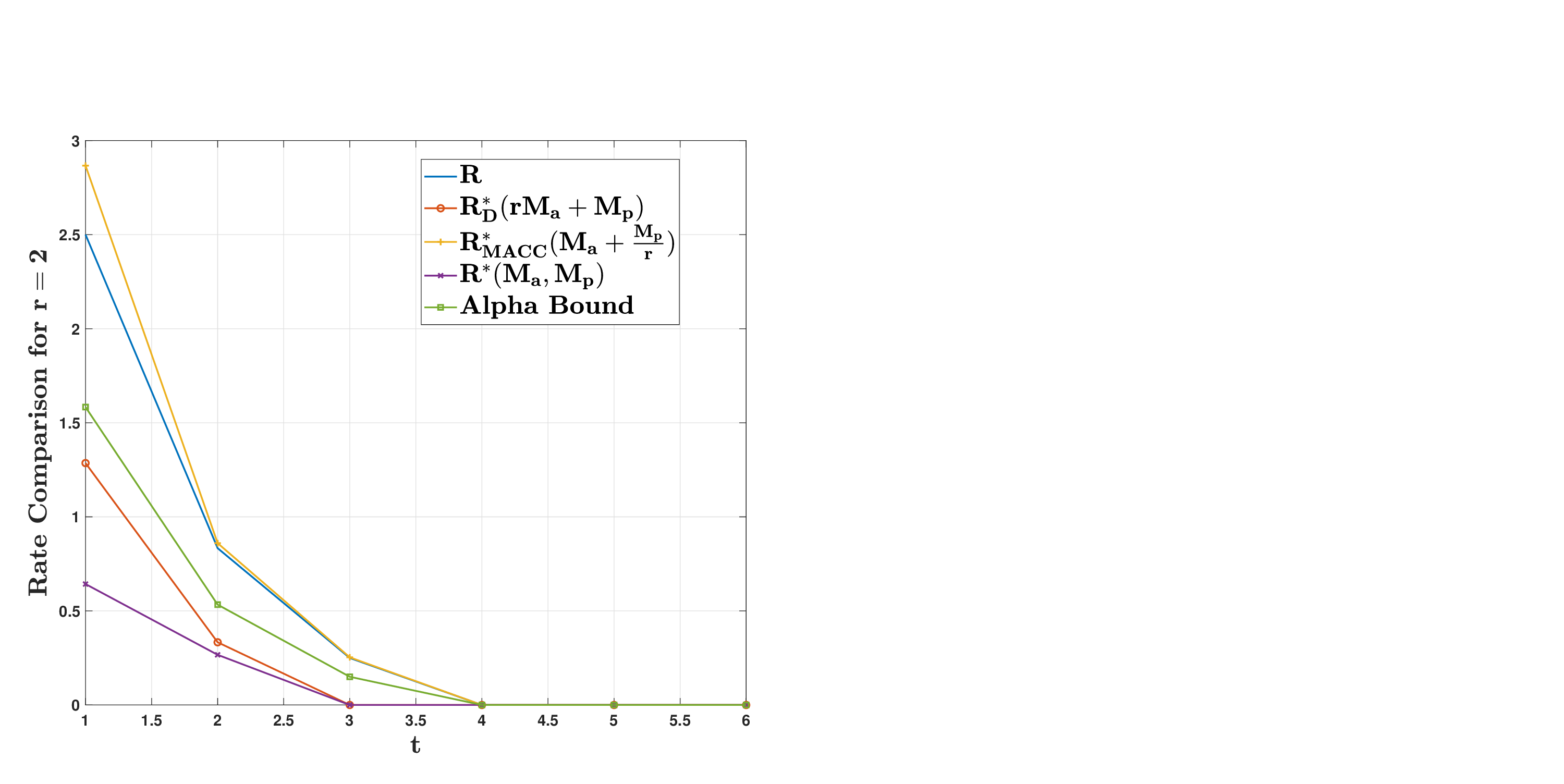}
			\caption{Rate vs. $t$ for $r=2$ and $M_p = \frac{N}{K}$.}
			\label{fig3}
		\end{figure}
		\begin{figure}
			\includegraphics[width=\textwidth,height=0.5\textwidth]{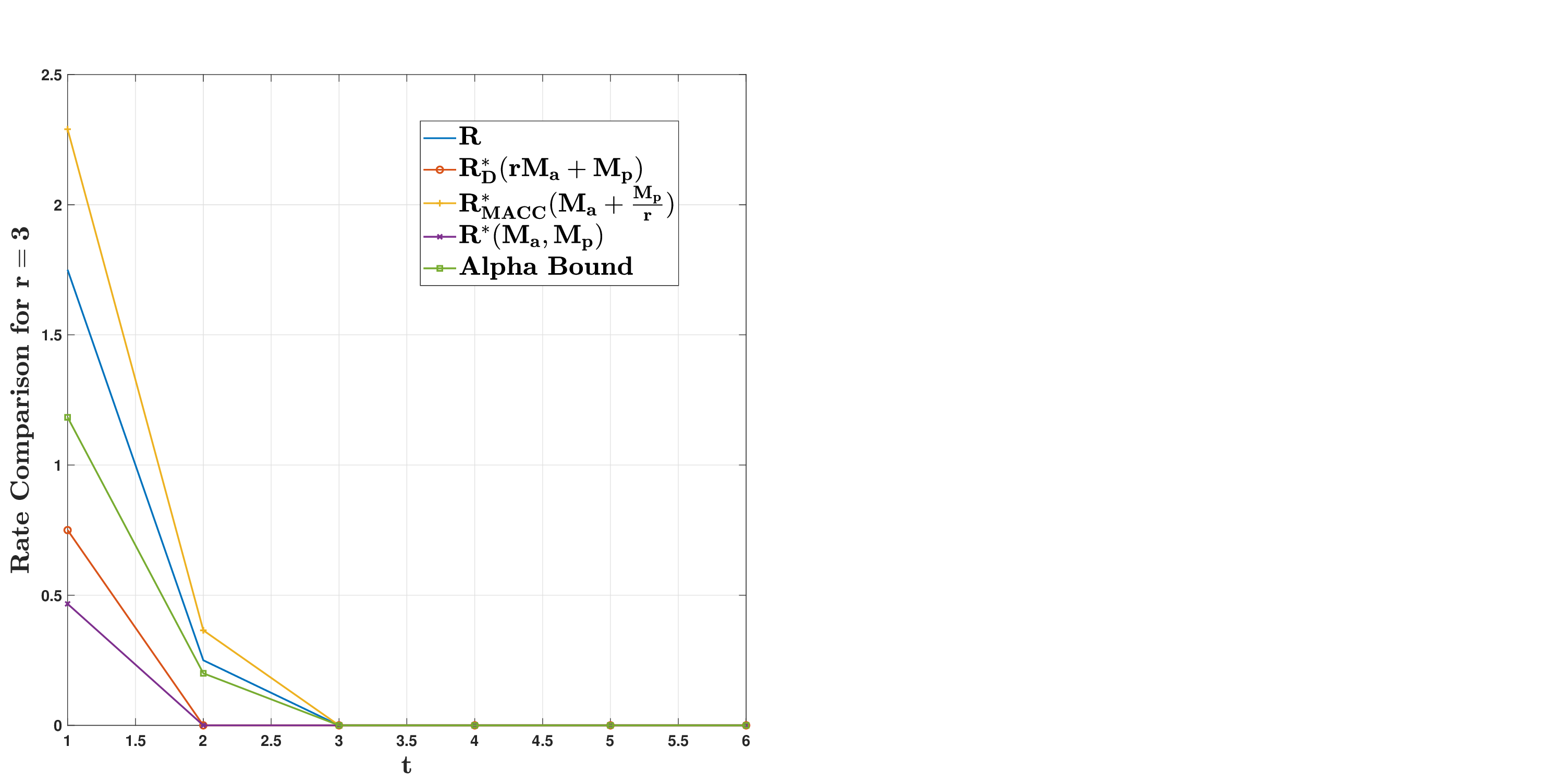}
			\caption{Rate vs. $t$ for $r=3$ and $M_p = \frac{N}{K}$.}
			\label{fig4}
		\end{figure}
		\begin{figure}
			%				\vspace{0.27cm}
			\includegraphics[width=\textwidth,height=0.46\textwidth]{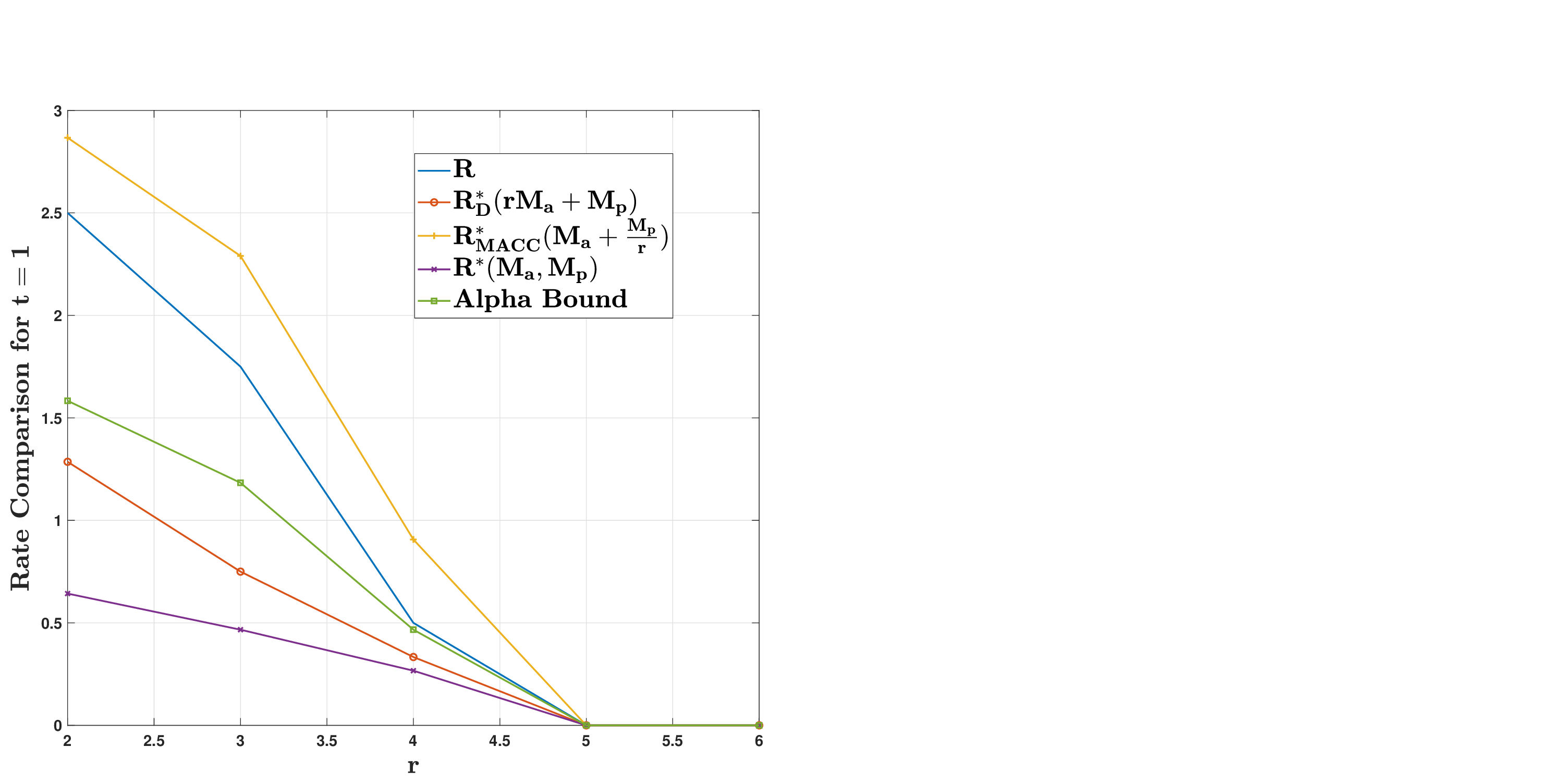}
			\caption{Rate vs. $r$ for $t=1$ and $M_p = \frac{N}{K}$.}
			\label{fig5}
		\end{figure}
		\begin{figure}
			\includegraphics[width=\textwidth,height=0.493\textwidth]{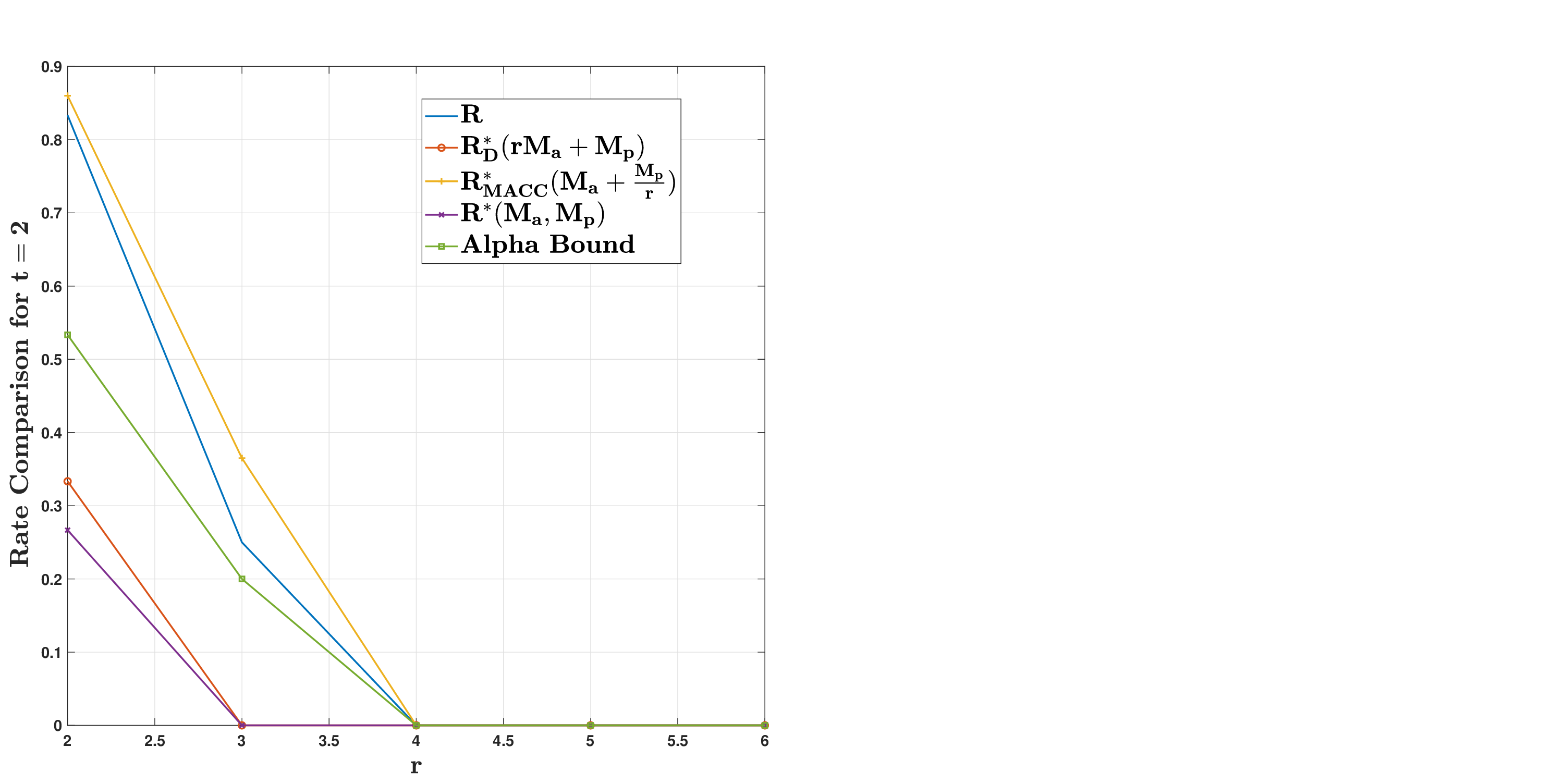}
			\caption{Rate vs. $r$ for $t=2$ and $M_p = \frac{N}{K}$.}
			\label{fig6}
		\end{figure}
		\section{Optimality for $\Lambda=4$ Case}
		\label{optimalityforlambda4}
		In this section, we examine the case of the CMAP coded caching system with $\Lambda=4$ access caches, exploring different combinations of access degree $r$, access cache memory $M_a$, and private cache memory $M_p$. We present a novel placement policy designed to accommodate multiple values of $M_p$, along with the corresponding transmissions made by the server for each combination of $M_a$, $r$, and $M_p$. 
		%Finally, we show that the server achieves the optimal worst-case rate for certain memory regimes.
		
		\textit{Placement Policy}: Each file is split into $\binom{\Lambda}{t_a}$ non-overlapping subfiles of equal size as shown:
		\begin{align*}
			W_n=\{W_{n,\mathcal{S}}:\mathcal{S}\subseteq[1,\Lambda], |\mathcal{S}|=t_a,\forall n\in[1,N]\},
		\end{align*}where $t_a=\frac{\Lambda M_a}{N}$, is the access cache replication factor. The contents of the access cache $i$ is 
		\begin{align}
			Z_i=\{W_{n,\mathcal{S}}:i\in\mathcal{S}, \mathcal{S}\subseteq[1,\Lambda], |\mathcal{S}|=t_a,\forall n\in[1,N]\},
		\end{align}for $i\in[1,\Lambda]$. Each access cache stores $\frac{N\binom{\Lambda-1}{t_a-1}}{\binom{\Lambda}{t_a}}=\frac{Nt_a}{\Lambda}=M_a$ files, satisfying its memory constraint. Note that the placement policy described above is the same as the placement policy described in section \ref{placementpolicy} for $t=t_a$.
		
		We will now describe how the private cache of the users are populated. The server populates the private cache of user $\mathcal{U}$ with mini-subfiles of the subfiles it does not get on connecting to access caches. Each subfile is further split into $\binom{\binom{\Lambda-t_a}{r}}{t_p}$ mini-subfiles, where $t_p=\frac{KM_p}{N}\in\mathbb{Z}^+$ is the private cache memory replication factor. $W_{n,\mathcal{S},\mathcal{U}_{i_1},\mathcal{U}_{i_2},\cdots,\mathcal{U}_{i_{t_p}}}$ denotes the mini-subfile of subfile $\mathcal{S}$ of file $n$ present in the private caches of users $\mathcal{U}_{i_1},\mathcal{U}_{i_2},\cdots,\mathcal{U}_{i_{t_p}}$, for some $i_1,i_2,\cdots, i_{t_p}\in[1,K]$. The contents of the private cache of user $\mathcal{U}$ is
		\begin{align}
			\label{placementprivatecachesoptimalitysection}
			&Z^p_{\mathcal{U}}=\{W_{n,\mathcal{S},\mathcal{T}_1,\mathcal{T}_2,\cdots,\mathcal{T}_{t_p}}:\mathcal{S}\subseteq[1,\Lambda]\setminus\mathcal{U},\nonumber\\&\{\mathcal{T}_i\in\{\mathcal{T}^\prime\subseteq[1,\Lambda]\setminus\mathcal{S},|\mathcal{T}^\prime|=r\},\forall i\in[2,t_p]\},\mathcal{T}_1=\mathcal{U},\nonumber\\&\forall n\in[1,N]\}.
		\end{align}The private cache of each user stores $\frac{N\binom{\Lambda-r}{t_a}\binom{\binom{\Lambda-t_a}{r}-1}{t_p-1}}{\binom{\Lambda}{t_a}\binom{\binom{\Lambda-ta}{r}}{t_p}}=\frac{N\binom{\Lambda-r}{t_a}t_p}{\binom{\Lambda}{t_a}\binom{\Lambda-t_a}{r}}=\frac{N\binom{\Lambda-r}{t_a}t_p}{\binom{\Lambda}{r}\binom{\Lambda-r}{t_a}}=\frac{Nt_p}{\binom{\Lambda}{r}}=M_p$ files, satisfying its memory constraint. Notably, for $t_p=1$, this placement reduces to the one discussed in Section \ref{placementpolicy}. We will now consider the case where $t_a\not\in\mathbb{Z}^+,t_p\not\in\mathbb{Z}^+$. Consider the following cases:
		\begin{enumerate}
			\item $t_a\not\in\mathbb{Z}^+,t_p\in\mathbb{Z}^+$. In this case, memory sharing is done in access caches, while no memory sharing is needed for the private caches.%This case is explained below:
			\begin{remark}
				Consider $M_a$ such that $t_a=\frac{\Lambda M_a}{N}$ is not an integer. Let $M_1=\frac{\lceil t_a\rceil N}{\Lambda}$ and $M_2=\frac{\lfloor t_a\rfloor N}{\Lambda}$. Since $M=\frac{t_aN}{\Lambda}$, we know that $M_2\leq M_a\leq M_1$. Hence, $M_a$ can be written as 
				\begin{align*}
					M_a={\alpha_1} M_1 + (1-{\alpha_1})M_2,
				\end{align*}for some $0\leq{\alpha_1}\leq 1$. The file $W_{n}$ is split into $W_{n}^{{\alpha_1}}$, of ${\alpha_1} B$ bits, and $W_{n}^{(1-{\alpha_1})}$, of $(1-{\alpha_1})B$ bits, respectively, $\forall n\in[1,N]$. The file $W_{n}^{\alpha_1}$ is further broken down into subfiles as $W_{n}^{\alpha_1}=\{W_{n,\mathcal{S}}^{\alpha_1}: \mathcal{S}\subseteq[1,\Lambda],|\mathcal{S}|=\lceil t_a\rceil\}$, while the file $W_{n}^{(1-{\alpha_1})}$ is broken into subfiles as $W_{n}^{(1-{\alpha_1})}=\{W_{n,\mathcal{S}}^{(1-{\alpha_1})}: \mathcal{S}\subseteq[1,\Lambda],|\mathcal{S}|=\lfloor t_a\rfloor\}$. The access caches are filled with subfiles $W_{n,\mathcal{S}}^{{\alpha_1}}$ as described in \eqref{placementaccess}, for $t=\lceil t_a\rceil$ and with subfiles $W_{n,\mathcal{S}}^{(1-{\alpha_1})}$, as described in \eqref{placementaccess}, for $t=\lfloor t_a\rfloor$. Thus, every access cache stores $N{\alpha_1}\binom{\Lambda-1}{\lceil t_a\rceil-1}B+N(1-{\alpha_1})\binom{\Lambda-1}{\lfloor t_a\rfloor-1}B$ bits, which is equivalent to $N{\alpha_1}\frac{\binom{\Lambda-1}{\lceil t_a\rceil-1}}{\binom{\Lambda}{\lceil t_a\rceil}}+N(1-{\alpha_1})\frac{\binom{\Lambda-1}{\lfloor t_a\rfloor-1}}{\binom{\Lambda}{\lfloor t_a\rfloor}}=\frac{N{\alpha_1}\lceil t_a\rceil}{\Lambda}+\frac{N(1-{\alpha_1})\lfloor t_a\rfloor}{\Lambda}={\alpha_1} M_1 +(1-{\alpha_1})M_2=M_a$ files, satisfying its memory constraint. 
				
				The private caches of the users will be populated with the mini-subfiles of $W_{n,\mathcal{S}}^{\alpha}$ and $W_{n,\mathcal{S}}^{(1-\alpha)}$ as described in \eqref{placementprivatecachesoptimalitysection} for $\lceil t_a\rceil$ and $\lfloor t_a \rfloor$. Every private cache stores $\frac{\alpha_1N\binom{\Lambda-r}{\lceil t_a\rceil}\binom{\binom{\Lambda-\lceil t_a\rceil}{r}-1}{t_p-1}}{\binom{\Lambda}{\lceil t_a\rceil}\binom{\binom{\Lambda-\lceil ta\rceil}{r}}{t_p}}+\frac{(1-\alpha_1)N\binom{\Lambda-r}{\lfloor t_a\rfloor}\binom{\binom{\Lambda-\lfloor t_a\rfloor}{r}-1}{t_p-1}}{\binom{\Lambda}{\lfloor t_a\rfloor}\binom{\binom{\Lambda-\lfloor ta\rfloor}{r}}{t_p}}=\frac{\alpha_1Nt_p}{\binom{\Lambda}{r}}+\frac{(1-\alpha_1)Nt_p}{\binom{\Lambda}{r}}=\alpha_1M_p+(1-\alpha_1)M_p=M_p$, satisfying its memory constraint. %The rate corresponding to $t=\lceil t\rceil$ is $\alpha_1 R_1$ and the rate corresponding to $t=\lfloor t\rfloor$ is $(1-\alpha_1)R_2$, respectively. Thus,% We will now explain how the rate at memory $M_a$ is calculated. We denote $R_{M_1}$ as the rate at memory point $M_1$, $R_{M_2}$ as the rate at memory point $M_2$, and, $R_{M_a}$ at memory point $M_a$. Rate $R_{M_a}$ is calculated as:
				%			\begin{align*}
					%				R_{M_a}=\alpha_1 R_{M_1}+(1-\alpha_1) R_{M_2}.
					%			\end{align*} 
			\end{remark}
			\item $t_a\in\mathbb{Z}^+,t_p\not\in\mathbb{Z}^+$. For this case, memory sharing is done for the private caches, and not for the access caches.
			\begin{remark}
				Consider $M_p$ such that $t_p=\frac{KM_p}{N}$ is not an integer, where $K=\binom{\Lambda}{r}$. Let $M_3=\frac{\lceil t_p\rceil N}{K}$ and $M_4=\frac{\lfloor t_p\rfloor N}{K}$. Since $M_p=\frac{t_pN}{K}$, we know that $M_4\leq M_a\leq M_3$. Hence, $M_p$ can be written as 
				\begin{align*}
					M_p=\alpha_2 M_3 + (1-\alpha_2)M_4,
				\end{align*}for some $0\leq\alpha_2\leq 1$. The subfile $W_{n,\mathcal{S}}$ is split into $W_{n,\mathcal{S}}^{\alpha_2}$, of $\alpha_2B_s$ bits, and $W_{n,\mathcal{S}}^{(1-\alpha_2)}$, of $(1-\alpha_2)B_s$ bits, respectively, where $B_s$ is the size of a subfile. The file $W_{n,\mathcal{S}}^{\alpha_2}$ is further broken down into mini-subfiles as $W_{n,\mathcal{S}}^{\alpha_2}=\{W_{n,\mathcal{S},\mathcal{T}_1,\mathcal{T}_2,\cdots,\mathcal{T}_{\lceil t_p\rceil}}^{\alpha_2}: \mathcal{T}_i\subseteq[1,\Lambda]\setminus S,|\mathcal{T}_i|=r,\forall i\in[1,\lceil t_p\rceil]\}$, while the subfile $W_{n,\mathcal{S}}^{(1-\alpha_2)}$ is broken into mini-subfiles as $W_{n,\mathcal{S}}^{(1-\alpha_2)}=\{W_{n,\mathcal{S},\mathcal{T}_1,\mathcal{T}_2,\cdots,\mathcal{T}_{\lfloor t_p\rfloor}}^{(1-\alpha_2)}: \mathcal{T}_i\subseteq[1,\Lambda]\setminus S,|\mathcal{T}_i|=r,\forall i\in[1,\lfloor t_p\rfloor]\}$. The private caches of the users are filled with the mini-subfiles as described in \eqref{placementprivatecachesoptimalitysection}. Thus, every private cache stores $N\alpha_2\frac{\binom{\Lambda-r}{t_a}\binom{\binom{\Lambda-t_a}{r}-1}{\lceil t_p\rceil-1}}{\binom{\Lambda}{t_a}\binom{\binom{\Lambda-ta}{r}}{\lceil t_p\rceil}}+N(1-\alpha_2)\frac{N\binom{\Lambda-r}{t_a}\binom{\binom{\Lambda-t_a}{r}-1}{\lfloor t_p\rfloor-1}}{\binom{\Lambda}{t_a}\binom{\binom{\Lambda-ta}{r}}{\lfloor t_p\rfloor}}=N\alpha_2\frac{\lceil t_p\rceil}{K}+N(1-\alpha_2)\frac{\lfloor t_p\rfloor}{K}=\alpha_2 M_3+(1-\alpha_2)M_4=M_p$ files, satisfying its memory constraint. %The rate corresponding to $t=\lceil t\rceil$ is $\alpha R_1$ and the rate corresponding to $t=\lfloor t\rfloor$ is $(1-\alpha)R_2$, respectively. Thus,% We will now explain how the rate at memory $M_a$ is calculated. We denote $R_{M_1}$ as the rate at memory point $M_1$, $R_{M_2}$ as the rate at memory point $M_2$, and, $R_{M_a}$ at memory point $M_a$. Rate $R_{M_a}$ is calculated as:
				%			\begin{align*}
					%				R_{M_a}=\alpha R_{M_1}+(1-\alpha) R_{M_2}.
					%			\end{align*} 
			\end{remark}
			\item $t_a\not\in\mathbb{Z}^+,t_p\not\in\mathbb{Z}^+$. In this case, memory sharing is done for both the private and the access caches, as explained above.
		\end{enumerate}

		We now examine all non-trivial combinations of $M_a, M_p$, and $r$. In other words, we focus on scenarios where a user can access only a part of the library. To calculate the mini-subfiles accessible to a user, we consider the subfiles obtained from the access and private caches. From the access caches, a user obtains $N\left[\binom{\Lambda}{t_a}-\binom{\Lambda-r}{t_a}\right]$, each of which yields all associated mini-subfiles. Thus, the total number of mini-subfiles obtained from the access caches are $N\left[\binom{\Lambda}{t_a}-\binom{\Lambda-r}{t_a}\right]\binom{\binom{\Lambda-t_a}{r}}{t_p}$. Additionally, a user receives $N\binom{\Lambda-r}{t_a}\binom{\binom{\Lambda-ta}{r}-1}{t_p-1}$ mini-subfiles from its private cache. Therefore, the total number of mini-subfiles accessible to a user is $N\left(\left[\binom{\Lambda}{t_a}-\binom{\Lambda-r}{t_a}\right]\binom{\binom{\Lambda-t_a}{r}}{t_p} + \binom{\Lambda-r}{t_a}\binom{\binom{\Lambda-ta}{r}-1}{t_p-1}\right)$. Since each file in broken down into $\binom{\Lambda}{t_a}\binom{\binom{\Lambda-t_a}{r}}{t_p}$ mini-subfiles, every user has access to $\frac{N\left(\left[\binom{\Lambda}{t_a}-\binom{\Lambda-r}{t_a}\right]\binom{\binom{\Lambda-t_a}{r}}{t_p} + \binom{\Lambda-r}{t_a}\binom{\binom{\Lambda-ta}{r}-1}{t_p-1}\right)}{\binom{\Lambda}{t_a}\binom{\binom{\Lambda-t_a}{r}}{t_p}}=N\left(1-\frac{\binom{\Lambda-ta}{r}-t_p}{\binom{\Lambda}{r}}\right)$ files. 
		
		Using the above calculations, for $\Lambda = 4$, the only non-trivial $(r,M_a,M_p)$ triplets are $(2,\frac{N}{4},\frac{N}{6})$ and $(2,\frac{N}{4},\frac{N}{3})$. For all the other points, the users have access to the entire library.We discuss optimality for these two cases for a CMAP system with  $\Lambda = 4$ access caches.  Note that the $(r=2,M_a=\frac{N}{4},M_p=\frac{N}{6})$ case corresponds to $r=2$, $t_a=1$, $t_p=1$ and the  $(r=2,M_a=\frac{N}{4},M_p=\frac{N}{3})$ triplet corresponds to the case where $r=2$, $t_a=1$, and, $t_p=2$.

%	$(\Lambda=4, N \geq\binom{\Lambda}{r}, r=2, t_a=1, t_p=1)$
			\subsection{Optimality for \protect{$(\Lambda=4, N \geq K, r=2, t_a=1, t_p=1)$}}
			This case has been discussed in Example \ref{example1}. The optimality is given for the case with $N \geq K  = \binom{\Lambda}{r}$ w.r.t the worst-case rate achieved. We begin this section by defining regular delivery schemes for coded caching systems.			
			\begin{defn}
				A delivery scheme for a coded caching system is said to be $g-$regular if each transmission in it is a coded combination of $g$ mini-subfiles.
			\end{defn}
			Note that for the CMAP coded caching system with $r=2$, $t_a=1$, and $t_p=1$, we have $\alpha \geq 5$. This means that the server will need to make at least five transmissions to satisfy the demands of all users.

			In the case considered, with subpacketization $F=12$, there are $K=6$ users, and each user has $8$ mini-subfiles. Therefore, each user demands $12-8=4$ mini-subfiles. Thus, the total number of mini-subfiles demanded by all the users together is $6 \times (12-4)=24$. Since five does not divide twenty-four, the minimum number of server transmissions in a regular delivery scheme is six. This has been achieved in Example \ref{example1}.
			
			Hence, the CMAP coded caching system with $\Lambda=4$ and parameters $r=2$, $t_a=1$, and $t_p=1$ is optimal w.r.t worst-case when $N \geq K$ under the given placement and regular-delivery scheme assumption.
			
				\subsection{Optimality for \protect{$(\Lambda=4, N \geq K, r=2, t_a=1, t_p=2)$}.}			
			Consider a CMAP coded caching system having a central server with a library of $6$ files, denoted as $\{W_1, W_2, W_3, W_4, W_5, W_6\}$ and a set of $\Lambda=4$ access caches, each capable of storing $M_a=1.5$ files. There are $K=6$ users, equipped with a private cache of capacity $M_a=2$ files, connecting to the system such that every user connects to a unique subset of $r=2$ access caches. Since $t_a=1$, each file $W_n$ is split into $\binom{\Lambda}{t_a}=4$ subfiles as $W_n=\{W_{n,1},W_{n,2},W_{n,3},W_{n,4}\}$ for $n\in[1,6]$. The server populates the access caches are shown below:
			\begin{align*}
				&Z_1=\{W_{n,1},\forall n\in[1,6]\},\\
				&Z_2=\{W_{n,2},\forall n\in[1,6]\},\\
				&Z_3=\{W_{n,3},\forall n\in[1,6]\},\text{ and},\\
				&Z_4=\{W_{n,4},\forall n\in[1,6]\}.
			\end{align*}Each access cache stores $\frac{6}{4}=1.5$ files, satisfying its memory constraint. Every subfile is broken down into $\binom{\binom{\Lambda-t_a}{r}}{t_p}=\binom{3}{2}=3$ mini-subfiles. Hence, the subpacketization is $F=12$. The server populates the private caches of the users with mini-subfiles of the subfiles the user does not obtain on connecting to access caches. The contents of the private caches of users are as shown below:
			\begin{align*}
				&Z^p_{12}=\{W_{n,3,12,14},W_{n,3,12,24},W_{n,4,12,13},W_{n,4,12,23}\},\\
				&Z^p_{13}=\{W_{n,2,13,14},W_{n,2,13,34},W_{n,4,12,13},W_{n,4,13,23}\},\\
				&Z^p_{14}=\{W_{n,2,13,14},W_{n,2,14,34},W_{n,3,12,14},W_{n,3,14,24}\},\\
				&Z^p_{23}=\{W_{n,1,23,24},W_{n,1,23,34},W_{n,4,12,23},W_{n,4,13,23}\},\\
				&Z^p_{24}=\{W_{n,1,23,24},W_{n,1,24,34},W_{n,3,12,24},W_{n,3,14,24}\},\text{ and},\\
				&Z^p_{34}=\{W_{n,1,23,34},W_{n,1,24,34},W_{n,2,13,34},W_{n,2,14,34}\},
			\end{align*}$\forall n\in[1,6].$ Each private cache stores $\frac{24}{12}=2$ files, satisfying its memory constraint. 
			
			The transmissions made by the server for the demand vector $\mathbf{d}=(d_{\mathcal{U}}:\mathcal{U}\subseteq[1,\Lambda],|\mathcal{U}|=r)$ are as presented below:
			%	\begin{align*}
			\begin{enumerate}
				\item $W_{d_{{12}},3,14,24}+W_{d_{{13}},2,14,34}+W_{d_{{14}},3,12,24}+W_{d_{{23}},1,24,34}+W_{d_{{24}},3,12,14}+W_{d_{{34}},1,23,24}.$
				\item $W_{d_{{12}},4,13,23}+W_{d_{{13}},4,13,23}+W_{d_{{14}},2,13,34}+W_{d_{{23}},4,12,13}+W_{d_{{24}},1,23,34}+W_{d_{{34}},2,13,14}.$
			\end{enumerate}
			It can be verified that the above transmissions are decodable and each of the six users recover the two missing mini-subfiles of their requested file from the above two transmissions. Since the server makes two transmissions, the rate $R=\frac{2}{12}=\frac{1}{6}$.
			\subsubsection{Discussion on Optimality}
			Consider the cut-set bound derived in Theorem \ref{thm3}. For $s=1$, we have $R^{\textasteriskcentered}(M_a,M_p)\geq 1-\frac{rM_a+M_p}{N}$. For the case discussed in this section, we have $R^{\textasteriskcentered}(1.5,2)\geq 1-\frac{2\frac{N}{4}+\frac{N}{3}}{N}=1-\frac{1}{2}-\frac{1}{3}=\frac{1}{6}$. The scheme is optimal for this case since it achieves the cut-set-based lower bound in Theorem \ref{thm1}.
			
			\section{Conclusions}
			\label{conclusions}
			We introduced the CMAP coded caching for which we provided an achievability scheme and characterized its rate. Further, we presented a lower bound on the number of transmissions for the proposed scheme using index coding techniques. We bounded the optimal worst-case rate under uncoded placement using the rates of the MAN schemes in \cite{MAN} and \cite{PD}. We showed using numerical comparisons  that the rate of the proposed scheme approaches the lower bound in certain memory regimes. For the special case when the CMAP system has four access caches, the optimality for all valid memory pairs were also discussed. 
			\section*{Acknowledgment}
			This work was supported partly by the Science and Engineering Research Board (SERB) of the Department of Science and Technology (DST), Government of India, through J.C Bose National Fellowship to Prof. B. Sundar Rajan.

		\end{document}